\documentclass[a4paper,twocolumn,11pt,accepted=2026-10-19]{quantumarticle}
\pdfoutput=1
\usepackage[utf8]{inputenc}
\usepackage[numbers]{natbib}
\usepackage[english]{babel}
\usepackage[T1]{fontenc}
\usepackage{amsmath}
\usepackage{hyperref}

\usepackage{tikz}
\usepackage{lipsum}
\usepackage{graphicx}% Include figure files
\usepackage{dcolumn}% Align table columns on decimal point
\usepackage{bm}% bold math
\usepackage{mathtools}
\usepackage{dsfont}
\usepackage{mathrsfs}
\usepackage{amsfonts,mathtools}
\usepackage{amsthm}
\usepackage{mathbbol}
\usepackage{xcolor}
\usepackage{amsmath}
\usepackage{algorithm}
\usepackage[noend]{algpseudocode}
\usepackage{framed}
\usepackage{algorithmicx}
\usepackage{mathabx}

\newcommand{\Sp}{\mathrm{Sp}}

\newcommand{\diag}{\mathrm{diag}}

\newcommand{\reals}{\mathds{R}}
\newcommand{\complex}{\mathds{C}}
\newtheorem{thm}{Theorem}

\newtheorem{defn}{Definition}
\newtheorem{lema}{Lemma}
\newtheorem{propos}{Proposition}

\newcommand{\ie}{\textit{i.e.,\ }}

\newcommand{\<}{\langle}
\renewcommand{\>}{\rangle}

\renewcommand{\phi}{\varphi}

\begin{document}

\title{Geometric characterization of non-Gaussian entanglement for finite stellar rank states}

\author{Carlos E. Lopetegui-Gonz\'alez}
\email{carlos-ernesto.lopetegui-gonzalez@lkb.upmc.fr}
\author{Massimo Frigerio}
\email{massimo.frigerio@lkb.upmc.fr}
\author{Mattia Walschaers}
\email{mattia.walschaers@lkb.upmc.fr}
\affiliation{Laboratoire Kastler Brossel, Sorbonne Universit\'{e}, CNRS, ENS-Universit\'{e} PSL,  Coll\`{e}ge de France, 4 place Jussieu, F-75252 Paris, France}

\maketitle

\begin{abstract}
  We introduce a general framework for the analysis of non-Gaussian entanglement in bosonic states of finite stellar rank. The central result is the full characterization of their entanglement structure through the atomic decomposition of their stellar polynomial and its associated structural graph, whose connected components determine the mode-intrinsic entanglement content of the state and all partitions compatible with passive separability. An essential ingredient in this construction is the concept of essential variables, which identify the minimal number of effective modes involved in a core state, in direct correspondence with the symplectic rank. This reduction provides the foundation for decomposing stellar polynomials into atomic factors and for revealing the underlying entanglement structure. Building on this, we derive complete separability criteria for two-mode states, expressed through hyperplane decompositions of zero sets, and for stellar-rank-2 states across an arbitrary number of modes. Applications to several example states illustrate how the method isolates genuinely non-Gaussian resources and quantifies preparation complexity. 
\end{abstract}
\section{Introduction}
Continuous-variable quantum optics \cite{braunstein_quantum_2005,Adesso_2014} provides one of the most promising pathways for scalability of quantum information and quantum computing. This is confirmed by the potential to create large entangled resource states in a deterministic way \cite{Larsen_2019,Asavanat_2019}. Yet, entanglement is not the only non-classical resource required for performing relevant information processing or computational tasks. Notably, it is well known that non-Gaussianity \cite{bartlett_efficient_2002}, even in the strong form of Wigner negativity \cite{mari_positive_2012,Veitch_2012}, is required to achieve quantum computational advantages with bosonic continuous-variable systems. It is thus from the intertwining between quantum entanglement \cite{horodecki_quantum_2009} and non-Gaussianity \cite{walschaers_non-gaussian_2021}, that the full potential of quantum information processing with bosonic systems stems. This makes the characterization of entanglement in non-Gaussian states a timely task, which turns out to be remarkably challenging. \par
The first problem to face when studying non-Gaussian entanglement as a resource is how to properly define it. So far, one common approach has been to refer to any entanglement that cannot be observed by relying on Gaussian criteria \cite{Simon2000,Duan2000} as \emph{non-Gaussian entanglement} \cite{barral_metrological_2024}. A different approach is to consider non-Gaussian entanglement to be the entanglement that persists in any choice of mode basis (i.e., under all passive linear-optical transformations). The reason is that such entanglement can only arise in non-Gaussian states, as will be further discussed later on. Recently, in \cite{chabaud_resources_2023} the relevance of this kind of entanglement for a certain class of sampling protocols was pointed out. Previous works \cite{walschaers_statistical_2017,sperling_mode-independent_2019,chabaud_holomorphic_2022,lopetegui_detection_2024} have also discussed the potential relevance of this kind of entanglement for quantum optics information processing. Different choices have been made for how to call this form of entanglement, including \textit{mode-independent entanglement, mode-intrinsic entanglement}, or just non-Gaussian entanglement. In what follows we will stick to the choice of calling it mode-intrinsic entanglement \cite{lopetegui_detection_2024}. Complementarily, we call passively separable all states that are not mode-intrinsic entangled.
%In the following, we will say that a bipartite state is \emph{passively separable} if it does not possess mode-intrinsic entanglement with respect to the chosen bipartition.  
\par

A further classification of the family of possible continuous-variable entangled states is provided by all states whose entanglement cannot be undone by Gaussian operations, \ie entanglement that cannot be undone by complementing passive linear optics with active operations. This set of states has recently been studied in \cite{zhao_genuine_2024} under the name of \emph{genuine non-Gaussian entanglement}. Passive operations being a subset of general Gaussian operations, the set of passively separable states is thus a subset of Gaussian-separable states, by which we refer to all bipartite states whose entanglement can be generated from a separable state by applying a general Gaussian unitary operation. This implies that all genuine non-Gaussian entangled states are also mode-intrinsic entangled states. Gaussian separability, and thus passive separability reduce the sample and time complexity of learning a classical description of a quantum state \cite{zhao_genuine_2024}.

\par
In this paper we focus our attention on the analysis of pure states. A discussion of the difficulties of addressing the more general case of mixed states is presented in the outlook section.

 \par
 The rest of the paper is structured as follows. In section \ref{sec:passive-Gaussian_sep} we formally define genuine non-Gaussian entanglement and mode-intrinsic entanglement. Following this, we briefly introduce the stellar representation of bosonic states \cite{chabaud_holomorphic_2022}, which can be used to draw conclusions about passive and Gaussian separability, based on properties of the core state of the given non-Gaussian state. In section \ref{sec:passive-sep_core-state}, we outline the connection between passive separability of core states and the existence of a factorization of their stellar polynomials on independent sets of variables. In section \ref{sec:structural_charact_stellar_poly} we develop a technique to probe this question, by introducing the atomic decomposition of stellar polynomials.  This decomposition 
%allows to associate a graph with each core state, 
encodes all the information about the intrinsic mode entanglement properties of the state. The procedure to construct the atomic decomposition is based on two main subroutines: factorization of polynomials into irreducible factors, and a dimensional reduction of a polynomial to its \emph{essential variables}.  A discussion of the relevance of both subroutines, for our analysis, and beyond, is also presented in section \ref{sec:structural_charact_stellar_poly}.
 
Finally, in section \ref{sec:special_cases}, we specialize our techniques to simpler but experimentally relevant cases, first tackling passive-separability in the two-mode setting for generic stellar rank and then restricting the stellar rank to $r=2$ but leaving the total number of modes unconstrained.
These special cases have been considered before, both in the context of mode-independent entanglement detection \cite{sperling_mode-independent_2019}, and of linear-optical conversion of multiphoton states \cite{Migdal_2014}, mainly focusing on the case of homogeneous polynomials, \ie states with a conserved photon number.
 
 \section{Passive and Gaussian separability}\label{sec:passive-Gaussian_sep}
 We begin this section by defining the set of passively separable states \cite{sperling_mode-independent_2019,walschaers_statistical_2017,chabaud_holomorphic_2022,lopetegui_detection_2024}. To give definitions in the full generality that are compatible with the results that we are going to present, we introduce them referring to generic $K$-partitions, although oftentimes the relevant questions concerning passive separability are just posed in terms of nontrivial \emph{bipartitions}. A $K$-partition $I_{K} = \{n_1,...,n_{K}\}$ of $M$ orthogonal modes\footnote{A set of $M$ orthogonal modes defines an orthonormal mode basis $\mathbf U=\{\mathbf u_j\in \mathbb C^{M} | j=\{1,\ldots,M \}, \quad \mathbf u_{j}^\dagger \mathbf u_k = \delta_{j,k}\}.$ } is a partition of the integer $M$ such that each integer $n_{j}$ in $I_{K}$ determines the number of orthogonal modes in each group of the partition: for example, the notation $I_{3} = \{ 5,3,1 \}$ describes a partition of $9$ orthogonal modes into a group of $5$ modes, a group of $3$ modes and one residual mode. Since the ordering of the modes is irrelevant when we are free to change the mode basis in any way, only the specification of the number of modes per partition is relevant. \\
\begin{figure}[htbp]
\centering
\includegraphics[width =0.7\linewidth]{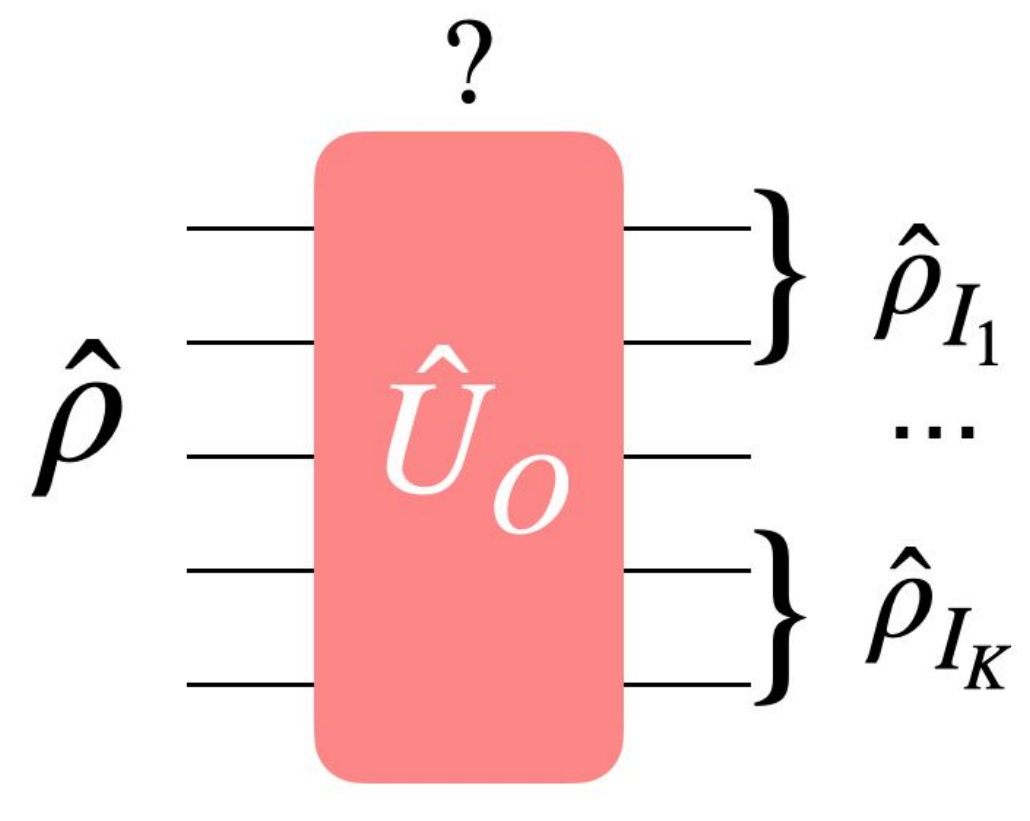}
\caption{A state $|\psi\rangle$ is passively separable with respect to some integer partition of the modes $\mathcal I_K$, if there exists a passive unitary $\hat U_O$, that disentangles it with respect to such a partition. An equivalent representation implies that the state can be generated from such a separable state by the implementation of the inverse of $\hat U_O$.}
\label{fig:passive_separability}
\end{figure}
 \begin{defn}\label{def:passive_sep}
     \textbf{Passively separable states} A bosonic quantum state $|\psi\rangle$ over $M$ modes is said to be passively separable with respect to some $K$-partition of the $M$ modes $I_{K} = \{n_1,...,n_K\}$ if there exists a passive linear-optical transformation ${\bf O} \in \mathcal K(M)= \mathrm{Sp}(2M,\mathbb R) \cap \mathrm{O}(2 M)$, that disentangles the state with respect to this $K$-partition, \ie 
     \begin{equation*}
         \hat{U}_{\mathbf{O}} |\psi\rangle =\bigotimes_{i=1}^{K} |\psi_{I_i}\rangle.
     \end{equation*}
 \end{defn}
 Here $\mathcal K(M)$ is the orthogonal compact subgroup of the real symplectic group $\mathrm{Sp} (2M,\mathbb R)$, given by its intersection with the orthogonal group $\mathrm{O}(2M)$, while $\hat{U}_\mathbf{O}$ is the unitary operator associated with the linear action of $\mathbf{O}$ on the phase-space. A schematics of this definition is provided in Fig.\ref{fig:passive_separability}.
 We will refer to the set of passively separable states with respect to some $K$-partition $I_K$ of $M$ modes as $\mathcal{P}_S(I_K,M)$. The set of mode-intrinsic entangled states is given by its complementary set. Formally we define it as \\
 \begin{defn}\label{def:mode-intrinsic_entanglement}
 \textbf{Mode-intrinsic entanglement}
 A state $|\psi\rangle$ over $M$ modes is mode-intrinsic entangled with respect to some $K$-partition $I_K$ if $|\psi\rangle\notin \mathcal{P}_S(I_K,M)$, \ie if there is no passive linear-optical operation that disentangles it with respect to the $K$-partition of interest.  
 \end{defn}
 \par 
 Notice the definition of mode-intrinsic entanglement is invariant under permutation of the modes, and thus under any relabeling of the $K$-partition $I_K$ that does not alter the number of modes per element of the partition. \par 
 \begin{figure}[htbp]
\centering
\includegraphics[width =0.7\linewidth]{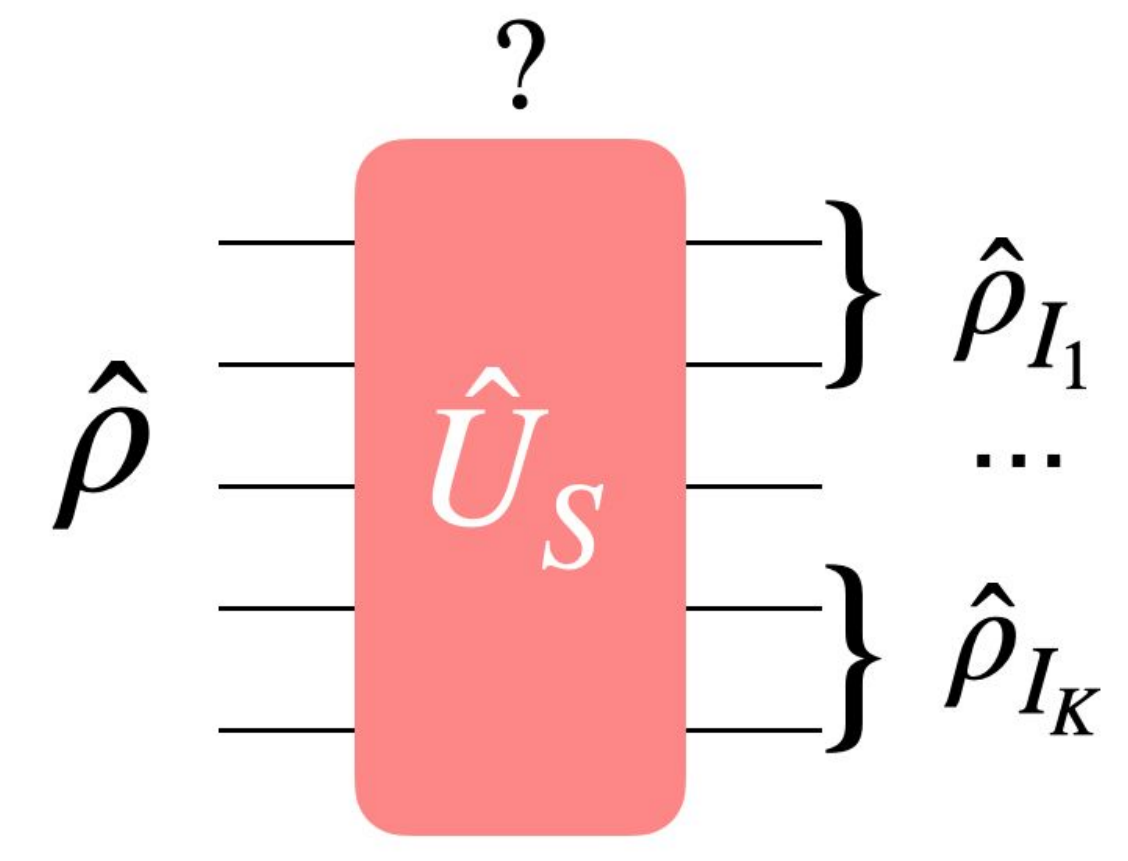}
\caption{A state $|\psi\rangle$ is Gaussian separable with respect to some integer partition of the modes $\mathcal I_K$, if there exists a Gaussian unitary $\hat U_S$, that disentangles it with respect to such a partition. }
\label{fig:Gaussian_separability}
\end{figure}
 \begin{defn}\label{def:Gaussian-separability}
     \textbf{Gaussian separable states} A bosonic quantum state $|\psi\rangle$ over $M$ modes is Gaussian separable with respect to some $K$-partition $I_K=\{I_1,...,I_K\}$, if there exists a Gaussian unitary transformation $\hat{U}_{\mathbf{S}}$, associated with a symplectic linear transformation $\mathbf{S} \in \Sp (2M,\mathbb R)$ acting on phase space, that disentangles the state with respect to $I_K$, \ie 
     \begin{equation*}
         \hat{U}_{\mathbf{S}} |\psi\rangle=\bigotimes_{i=1}^{K}|\psi_{I_i}\rangle.
     \end{equation*}
 \end{defn}
 A schematics of this definition is provided in Fig.\ref{fig:Gaussian_separability}. Complementary to the set of Gaussian separable states, which we will refer to as $\mathcal{G}_{S}(I_K,M)$, we define the set of genuinely non-Gaussian entangled states \cite{zhao_genuine_2024} as\\ 
 \begin{defn}\label{def:genuine_non-Gaussian_ent}
     \textbf{Genuine non-Gaussian entanglement} A bosonic quantum state $|\psi\rangle$ over $M$ modes is said to be genuinely non-Gaussian entangled with respect to some $K$-partition $I_K$ if $|\psi\rangle\notin \mathcal{G}_{S}(I_K,M)$, \ie if there is no symplectic transformation that disentangles it with respect to the partition $I_K$.
 \end{defn}\par
 The same invariance under permutations that applies to passive separability applies to the case of Gaussian separability. Notice that, because the set of passive linear-optical transformations $\mathcal K(n)$ is a subset of the set of Gaussian (symplectic) operations $\mathcal K(n)\subset \Sp(2n,\mathbb R)$, the set of passively separable states is a subset of the Gaussian separable state, $\mathcal{P}_S(I_K,M)\subset \mathcal{G}_{S}(I_K,M)$. Complementarily, the set of mode-intrinsic entangled states ($\overline{\mathcal{P}}_S$), includes the set of genuinely non-Gaussian entangled states ($\overline{\mathcal{G}}_{S}$), \ie $\overline{\mathcal{G}}_{S}(I_K,M)\subset \overline{\mathcal{P}}_{S}(I_K,M) $. \par

 \subsection{The stellar rank formalism}
 In the following we briefly introduce the stellar formalism \cite{chabaud_holomorphic_2022}, a very convenient framework for the analysis of the different types of entanglement we have introduced above. The stellar function representation of a quantum state is defined by \cite{chabaud_holomorphic_2022}: 
\begin{equation}\label{eq:stellar_function}
     F^{\star}_{\psi}(\mathbf z)=\exp(\frac{1}{2}\left|\mathbf z\right|^2)\<\mathbf z^{*}|\psi\>,
\end{equation}
where $\mathbf z\in \complex^{2M}$, and $|\mathbf z^{*}\rangle$ is the coherent state of amplitude $\mathbf z^{*}$. The stellar function $F^{\star}_{\psi}$ is thus given by the overlap between a coherent state and the quantum state $|\psi\rangle$, and can be thought of as a \textit{wavefunction} in phase space, given its relation to the Husimi Q-function:
\begin{equation*}
    Q_{\psi}(\mathbf{z})=\frac{e^{-|\mathbf{z}|^2}}{\pi^M} \mathcal \vert F_{\psi}^{\star}(\mathbf{z})\vert^2. 
\end{equation*}
When $M=1$, the number of zeros of $F^{\star}_{\psi}$ is called the \emph{stellar rank} of $\vert \psi \rangle$ and it is $0$ if and only if the state is Gaussian, while for non-Gaussian states it can even be infinite. For $M \geq 2$, an analogous definition of stellar rank holds, by counting the total degree of the stellar function.

 \begin{thm}\label{Stellar representation}
 \textbf{Stellar representation of bosonic quantum states}
 Every generic pure state $\vert \psi \rangle$ of a bosonic quantum system over $M$ modes can be represented as $|\psi\rangle=\hat{U}_{\mathbf{S}_\psi} |C_{\psi}\rangle$, where $\mathbf{S}_{\psi}\in \mathrm{Sp}(2M,\mathbb R)$ is the symplectic linear transformation associated with the Gaussian unitary $\hat{U}_{\mathbf{S}_\psi}$, and $|C_{\psi}\rangle$ is a core state, i.e. a superposition of multimode Fock states of up to $r$ photons, $|C_{\psi}\rangle=\sum_{\mathbf n|\sum_i n_i\leq r}c_{\mathbf n}|n_1,...,n_M\rangle$, where $r$ is the stellar rank of the state $|\psi\rangle$.
 \end{thm}
 This result neatly decouples the non-Gaussian properties of a state from the Gaussian effects, making the analysis of mode-intrinsic and genuine non-Gaussian entanglement quite straightforward. Notice, however, that the core state $\vert C_\psi \rangle$ associated with a pure non-Gaussian state $\vert \psi \rangle$ is defined up to a passive Gaussian unitary transformation on the $M$ modes, since these transformations send core states into core states preserving the stellar rank. Moreover, while $\hat{U}_{\mathbf{S}_\psi}$ uniquely defines the symplectic transformation $\mathbf{S}_\psi$, in general the unitary transformation might also include a multimode displacement, which acts as an affine translation on phase space, on top of the linear action of $\mathbf{S}_\psi$. Since displacements are local in every mode basis and they never generate entanglement, we will overlook them in the following. 
 
 Below we present two previously proven results \cite{chabaud_holomorphic_2022} concerning, respectively, the passive and the Gaussian separability of non-Gaussian states. In our context, they serve the main purpose of reducing the separability properties of generic non-Gaussian states to the passive separability of their core states. 
 \begin{lema}\label{lema:passive_separability}
    \textbf{Passive separability of pure non-Gaussian states } 
    A bosonic quantum state $|\psi\rangle$ with stellar representation $|\psi\rangle= \hat{U}_{\mathbf{S}_\psi} |C_{\psi}\rangle $ is \emph{passively separable} with respect to a $K$-partition $I_K$, if and only if its core state is \emph{separable} with respect to $I_K$ in a basis in which the squeezing operations in $\hat{U}_{\mathbf{S}_\psi}$ are local with respect to $I_K$. Equivalently:  
    \begin{equation*}
\hat{U}_{\mathbf{O}_1}|C_\psi\rangle=\bigotimes_{i=1}^{K}|C_{\psi}^{(I_i)}\rangle.
    \end{equation*}
    where $\mathbf{O}_1$ is related to the Bloch-Messiah decomposition of $\mathbf{S}_\psi$ by:
    \[ \mathbf{S}_\psi = \mathbf{O}_2 \left[ \bigoplus_{j=1}^{M} \diag(e^{2 r_j }, e^{-2r_j}) \right] \mathbf{O}_1 \]
    In that case, any passive linear-optical operation of the form $\left(\bigotimes_{i=1}^{K}\hat{U}_{\mathbf{O}^{'}_{I_{i}}} \right) \hat{U}_{\mathbf{O}_2}^{\dagger}$ disentangles the state with respect to the $K$-partition $I_K$, where $\mathbf{O}^{'}_{I_i}$ are arbitrary orthogonal symplectic transformations involving only modes $I_i$. 
 \end{lema}

 \begin{lema}\label{theorem:genuine_non-Gaussian_ent}
     \textbf{Gaussian separability of pure non-Gaussian states} A finite stellar rank bosonic quantum state $|\psi\rangle$, with stellar representation given by $|\psi\rangle=\hat{U}_{\mathbf{S}_\psi}|C_{\psi}\rangle$, is \emph{Gaussian separable} with respect to a $K$-partition $I_K$ if and only if its core state $|C_\psi\rangle$ is \emph{passively separable} with respect to the same partition. If $\hat{U}_{\mathbf{O}_{C_\psi}}$ is a passive unitary Gaussian transformation that disentangles $|C_\psi \rangle$, then a unitary Gaussian transformation that disentangles the state $|\psi\rangle$ is given by $\hat{U}_{\mathbf{O}_{C_\psi}} \hat{U}_{\mathbf{S}_\psi}^{\dagger}$.
 \end{lema}
 The proof of this claim is a direct combination of the definition of passive separability and the stellar decomposition and can be found on \cite{chabaud_holomorphic_2022}. Both of the results presented above highlight the relevance of the analysis of the entanglement properties of core states for the analysis of non-Gaussian entanglement. In what follows we will focus on the analysis of the passive separability of core states, \ie states with a finite support on Fock space. In passing, we highlight that such a task is also of particular relevance for single-photon based platforms for optical quantum information processing \cite{knill_scheme_2001}.

\section{Passive separability of core states}\label{sec:passive-sep_core-state}
In this section we analyze the passive separability of pure core states. This problem can be related to the analysis of factorization properties of polynomials. Indeed, given a core state with stellar rank $r$:
\begin{equation*}
    |C\rangle=\sum_{\mathbf n|\sum n_i\leq r} C_{\mathbf n}|n_1,...,n_M\rangle,
\end{equation*}
its stellar function is a complex polynomial, which we will refer to as the stellar polynomial:
\begin{equation}\label{eq:stellar_polynomial}
    p_{C}(\mathbf z)=\sum_{\mathbf n|\sum n_i\leq r}\frac{C_{\mathbf n}}{\sqrt{n_1! ...n_M!}} z_1^{n_1}...z_M^{n_M}. 
\end{equation}
As a consequence, we can analyze the entanglement of the quantum state $|C\rangle$ by studying the factorizations of its stellar polynomial. \par 
\begin{figure}[htbp]
\centering
\includegraphics[width =0.5\linewidth]{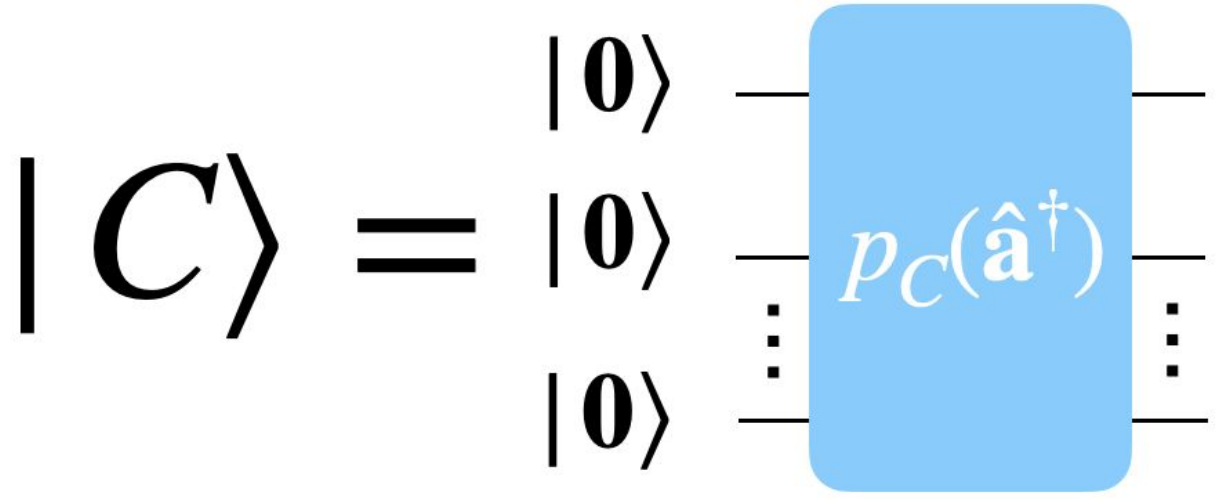}
\caption{Every core state can be obtained from the vacuum state by the implementation of the corresponding stellar polynomial of the creation operators over all the populated modes.  }
\label{fig:stellar_poly}
\end{figure}
Moreover, the stellar polynomial provides a proxy to analyze the preparation of the state $|C\rangle$. The latter can be obtained from the vacuum state $|C\rangle=p_C(\hat{\mathbf a}^\dagger)|\mathbf 0\rangle$\footnote{red}{$\hat{\mathbf a}^\dagger=\{\hat a_1^\dagger,\ldots, \hat a_M^\dagger\}$ represents the vector of the creation operators corresponding to all $M$ modes in a given mode basis.}, \ie by implementing the corresponding polynomial of creation operators, as shown in Fig.\ref{fig:stellar_poly}. Such implementation is always possible through some customized boson sampling post-selected schemes \cite{andrei2025}. 
\begin{propos}\label{propos:entanglement_vs_polynomial_factorization}
    A core state $|C\rangle$, of stellar rank $r$, over $M$ modes, is separable with respect to a $K$-partition $I_K$, if and only if its stellar polynomial is factorizable as 
    \begin{equation*}
        p_{C}(\mathbf z)=\prod_{i=1}^{K} p_{i}(\mathbf z_{I_i}),
    \end{equation*}
    where $z_{I_i}\in \complex^{|I_{i}|}$, are the phase space coordinates over the modes in $I_i$. The following equality holds $\sum_{i=1}^{K} \text{deg}(p_i)=r$, where $\text{deg}(p_i)$ is the degree of the polynomial $p_i$. 
\end{propos}
\begin{proof}
   This statement is a consequence of the one-to-one mapping between stellar polynomials and core states. To each polynomial factor $p_i(\mathbf z_{I_i})$ we can associate a core state $|C_{I_i}\rangle$, which implies $|C\rangle=\bigotimes_{i=1}^{K}|C_i\rangle$. Notice that if two factors share some variable the implication no longer holds.   
\end{proof}
This statement allows us to further abstract the analysis of the passive separability of core states to the analysis of geometric and algebraic properties of their stellar polynomials. As a phase space representation of quantum states, the stellar function transforms in a very neat way under passive linear optics: 
\begin{equation}\label{eq:LO_transform_core_states}
    p_{\hat{U}_{\mathbf{O}} |C\rangle}(\mathbf z)=p_{C}(\mathbf{U} \mathbf z),
\end{equation}
 where $\mathbf{U} \in \mathrm{U}(M)$ is the unitary $M \times M$ matrix describing the action of the real symplectic transformation ${\mathbf{O}}\in \mathcal K(M)$ on the complex coordinates $\mathbf z\in \complex^{M}$. This effect amounts to a rotation without rescaling of the axes in phase space. This implies the following statement:
\begin{propos}\label{propos:passive_sep_and_poly_factorization}
    A core state $|C\rangle$ over $M$ modes, whose stellar polynomial is $p_{C}(\mathbf z)$, is passively separable with respect to a $K$-partition $I_K$  if and only if there exists a unitary transformation $\mathbf{U}\in \mathrm{U}(M)$, such that 
    \begin{equation*}
        p_{C}(\mathbf{U} \mathbf z)=\prod_{i=1}^{K} p_{i}(\mathbf z_{I_i}).
    \end{equation*}
\end{propos}
It is important to highlight that $\mathbf{U}$, as a finite dimensional unitary matrix, is a distinct object from $\hat{U}_{\mathbf{O}}$, the unitary operator acting on the Hilbert space; they are only related through the isomorphism between $\mathcal{K}(M)$ and $\mathrm{U}(M)$ mapping the real phase-space $\reals^{2M}$ to the complex one $\complex^M$ and sending $\mathbf{O} \mapsto \mathbf{U}$. Since the same matrix $\mathbf{U}$ transforms the annihilation operators of the different modes in the Heisenberg picture and each mode $\hat{a}_j$ corresponds to a variable $z_j$ in the stellar representation, we can view its action either as a linear unitary change of variables of the stellar polynomial or as a mode basis change for the core state; correspondingly, selecting an orthonormal basis of vectors in $\complex^M$ fixes a choice of independent modes and variables and we shall heavily rely on this viewpoint in the following. The analysis of the emergence of passive separability of core states is thus equivalent to the analysis of the factorizability of polynomials in terms of independent sets of variables, a nontrivial problem that we tackle by leveraging techniques developed in the following section. \par

\section{Structural characterization of stellar polynomials}
\label{sec:structural_charact_stellar_poly}
In this section we present a systematic way to determine whether a multi-mode core state is passively separable with respect to some $K$-partition $I_K$ of the $M$ modes. %showing that each core state (and, consequently, its corresponding stellar polynomial), can be put in a fundamental form, through a change of mode basis, in which it is factorized into non-passively separable states in such a way that it has the least amount of multipartite mode entanglement with respect to any other mode basis. 
There are two integer parameters that jointly increase the complexity of the problem: the stellar rank of the state and the number of modes. The former increases the degree, and the latter the number of variables, of the stellar polynomial of the state.
%which we consider to investigate passive separability according to proposition \ref{propos:passive_sep_and_poly_factorization}. 
\par
\subsection{Essential variables of stellar polynomials}
As for the number of variables involved in the polynomial, it is natural to wonder whether it can be reduced. If a core state is not presented in the most convenient mode basis, it may be possible that it effectively populates a smaller number of modes. The problem of finding the minimum number of modes populated by a given state has previously been addressed, for example, by considering the coherence matrix of a state \cite{Treps_2020}. Moreover, the effective number of modes of a core state is equal to the symplectic rank of non-Gaussian quantum states \cite{mele_symplectic_2025}, which has recently been proven to be a relevant resource for quantum computational advantage with bosonic systems. For our purposes, it is convenient to compute the effective number of modes populated by a core state as the number of essential variables of the corresponding stellar polynomial, according to the following definition:
\begin{defn}
     Consider a polynomial $p(\mathbf{z})$ of degree $r$ in $M$ complex variables $\mathbf{z} \in \complex^M$. The \emph{space of essential variables} of $p$ is the vector space $\mathbb{E}(p)$ defined as:
    \begin{equation}
        \mathbb{E}(p) \ = \ \mathrm{span} \left\{ \vec{\nabla} p (\mathbf{z}) \vert \ \ \mathbf{z} \in \complex^M \right\} \ \subseteq \ \complex^M
    \end{equation}
    where the span is taken over all possible pointwise evaluations of the vector field $\vec{\nabla} p: \complex^M \to \complex^M$. Given any $\mathbf{v} \in \mathbb{E}(p)$, the linear combination $\mathbf{v}^T \mathbf{z}$ is called an \emph{essential variable} of $p(\mathbf{z})$ whereas, given $\mathbf{w} \in \mathbb{E}^{\perp}(p)$, the linear combination $\mathbf{w}^T \mathbf{z}$ is called an \emph{irrelevant variable} with respect to $p(\mathbf{z})$.
\end{defn}
In the state picture, this implies that the state, in the mode basis in which it was given, can be obtained as described in Fig.\ref{fig:dimentional reduction}, by first applying a polynomial of creation operators over the effective modes and then performing an inexpensive mode basis change, which may reduce considerably the complexity of the state preparation. \par
\begin{figure}[htbp]
\centering
\includegraphics[width =0.95\linewidth]{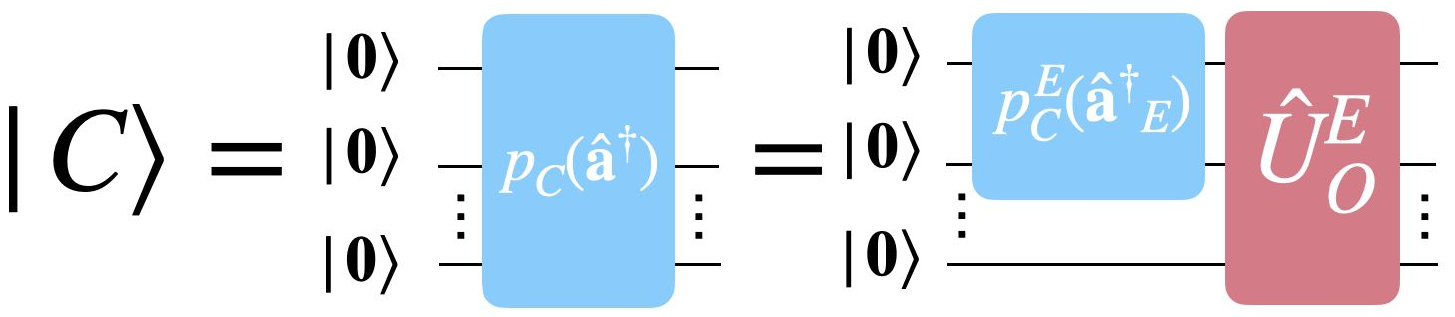}
\caption{Finding the essential variables of the stellar polynomial $p_C(\mathbf z)$, effectively decreases both the complexity of the state preparation and simplifies the analysis of the mode-intrinsic entanglement properties of the state. The modes on the vacuum (those that are not affected by the polynomial of creation and annihilation operators $p_C^E(\mathbf{\hat a_E^{\dagger}})$), are irrelevant for the analysis of the passive separability.   }
\label{fig:dimentional reduction}
\end{figure}
We will say that a multivariate polynomial $p(\mathbf{z})$ is \emph{reduced to its essential variables} if all of its variables $z_1,\dots,z_M \in \complex$ are essential. The following result, already known in algebraic geometry \cite{Carlini_2006}, establishes that $\mathbb{E}(p)$ defines the smallest space of variables upon which $p$ intrinsically depends in a nontrivial way, thereby motivating the nomenclature of "essential variables".
\begin{thm}
    The space of essential variables of a stellar polynomial $p(\mathbf{z})$ of degree $r$ in $M$ complex variables is covariant under linear unitary transformations of its space of variables. In particular, its dimension $M_E(p) = \dim \mathbb{E}(p)$ is invariant and equal to the minimum number of variables that appear with degree $\geq 1$ in $p(\mathbf{U} \mathbf{z})$ under any possible unitary linear transformation $\mathbf{U} \in \mathrm{U}(M)$ of the variables. Moreover, if $\mathbf{w} \in \complex^M$ is in the orthogonal complement $\mathbb{E}^{\perp}(p)$ of $\mathbb{E}(p)$, then $p$ is independent on the linear combination of variables $\mathbf{w}^T \mathbf{z}$. Correspondingly, $\mathbb{E}^{\perp}(p)$ identifies the space of modes that are in the vacuum with respect to the core state associated with $p$.
 \begin{proof}
 Let $\mathbf{U}: \complex^M \to \complex^M$ be any linear unitary transformation of the variables and denote by $(p \circ \mathbf{U})(\mathbf{z}) = p(\mathbf{U} \mathbf{z})$ the polynomial in the transformed variables. Then:
 \begin{equation}
\vec{\nabla}[p \circ \mathbf{U}] = [\mathbf{U}^T \vec{\nabla}p] \circ \mathbf{U}
 \end{equation}
Since $\mathbf{U}$ is invertible and, in constructing $\mathbb{E}(p)$, we take every possible $\mathbf{z} \in \complex^M$, the overall effect of the unitary change of variables on the space of essential variables is to covariantly rotate every vector in $\mathbb{E}(p)$ by $\mathbf{U}^T$, preserving its dimension. Also, if $\mathbf{w} \in \mathbb{E}^{\perp}(p)$ defines an irrelevant variable for $p$, we have that $\forall \mathbf{z} \in \complex^M: [\mathbf{w}^T \vec{\nabla}p](\mathbf{z}) = 0$, i.e. $p$ is everywhere constant along the direction $\mathbf{w}$, hence it does not depend on the irrelevant variable $\mathbf{w}^T \mathbf{z}$. Moreover, it is immediate to check that deriving a stellar polynomial of a core state with respect to $z_j$ amounts to acting with $\hat{a}_j$ on the core state itself; therefore, if $\mathbf{w}^T \mathbf{z}$ is an irrelevant variable, the core state is annihilated by $\mathbf{w}^T \vec{\hat{a}}$, meaning that the corresponding mode is in the vacuum state. Finally, notice that the converse is also true: $p$ constant along $\mathbf{w} \implies \mathbf{w} \in \mathbb{E}^{\perp}(p)$, therefore $p$ cannot be constant along any vector in $\mathbb{E}(p)$. This fact, together with the covariance of the space of essential variables, implies that $p$ can never depend on less than $\dim \mathbb{E}(p)$ variables under any unitary change of variables; in particular, if the initial space of variables, isomorphic to $\complex^M$, is partitioned into essential and irrelevant variables ($\mathbb{E}(p) \oplus \mathbb{E}^{\perp}(p)$) by introducing a basis for the two spaces, $p$ will precisely depend only on the first $\dim \mathbb{E}(p)$ new variables. 

\end{proof}
\end{thm}

As a side note, we notice that an efficient way to compute $\mathbb{E}(p)$ is to first find $\mathbb{E}^{\perp}(p)$ as the kernel of the matrix $\mathbf{G}$, known as the \emph{catalecticant} of $p$ in algebraic geometry, which is defined as:
 \begin{equation}
        \mathbf{G} \ := ( \mathrm{coeff}[\vec{\nabla} p]_1 , \dots , \mathrm{coeff}[\vec{\nabla} p]_M )
    \end{equation}
where $\mathrm{coeff}[\vec{\nabla}p]_{j}$, the column $j$ of $\mathbf{G}$, is the list of coefficients of the monomials in $\frac{\partial p}{\partial z_j}$ according to some arbitrary but fixed order\footnote{Notice that $\mathbf{G}$ is an $N \times M$ matrix with $N = \binom{M+r-1}{r-1}$ being the number of coefficients in a polynomial of order $r-1$ in $M$ variables.}. Clearly, any vector of $\complex^M$ in $\ker(\mathbf{G})$ is orthogonal to $\vec{\nabla}p(\mathbf{z})$ at all points $\mathbf{z}\in \complex^M$ and since the kernel is a vector space, we have $\ker(\mathbf{G})=\mathbb{E}^{\perp}(p)$. This method has a twofold role in our construction:  
on the one hand, it is an efficient tool to reduce the dimensionality of the problem we have to deal with; on the other hand, it provides a powerful subroutine to investigate the non-Gaussian entanglement structure of a core state, since it can be used to tell if two stellar polynomials share an essential variable or not. \par

\par
\subsection{Factorization into irreducible polynomials}\label{sec:irred_factorization}
In this subsection we present the last piece required to build up a scheme for determining the passive separability of multimode pure core states, from the analysis of its stellar polynomial. Any polynomial over any field can be uniquely factorized into a product of irreducible polynomials of lower order \cite{Sharpe_1987,Kopylov_2025}, akin to a prime number factorization of natural numbers, 
\begin{equation}
    p(\mathbf z)=\prod_{k=1}^{F} p_{k}^{n_k}(\mathbf z),
\end{equation}
where $p_k$ is an irreducible polynomial, \ie a polynomial that admits no further factorization, whose degree will be denoted by $r_k$ such that $\sum_k n_k r_k=r$, where $n_k$ is the power with which $p_k$ appears in the factorization and $r$ is the total degree of the stellar polynomial. For polynomials in a single variable and homogeneous polynomials in two variables, the irreducible factors are always linear polynomials ($\forall k \in \{1,\dots,F\}: r_k =1$), while irreducible polynomials in 3 or more variables can be nonlinear.\par

\begin{figure}[htbp]
\centering
\includegraphics[width =0.95\linewidth]{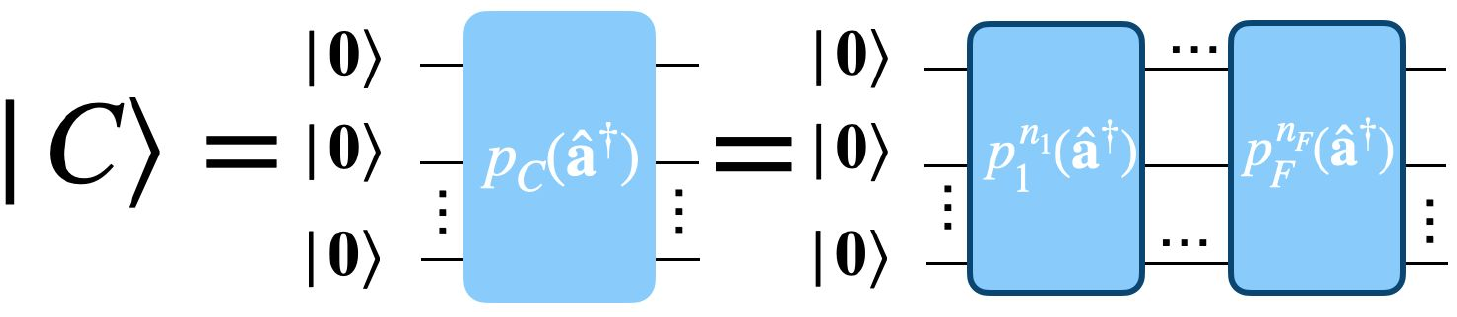}
\caption{The irreducible factorization of the stellar polynomial of a core state $|C\rangle$, breaks down the state preparation into a sequential application of all the irreducible factors. Moreover it simplifies the analysis of the mode-intrinsic entanglement to the analysis of the geometrical relations among them.   }
\label{fig:irreducible_factorization}
\end{figure}
Factorizing arbitrary complex polynomials over many variables is a demanding computational task. To begin with, exact factorization is an unstable numerical problem, since small perturbations of the coefficients are enough to render a polynomial irreducible \cite{gao_approximate_2004}. Nevertheless, there are several algorithms that perform this task for exact input coefficients \cite{Koepf2021}, in polynomial time of the input size. This means that, in general, the complexity has a polynomial dependence on the degree and exponential on the number of variables. Interestingly, if the multivariate polynomial is sparse on the monomial basis, there is no way in general to exploit the reduced dimension of the input to speed up the factorization \cite{Kaltofen2005}, except under specific conditions, like low degree. \par
On several symbolic programming languages like Sage \cite{sagemath} or Mathematica \cite{Mathematica}, factorization tools are available for generic multivariate polynomials with complex coefficients. On the other hand, there are some algorithms for factorization with inexact coefficients \cite{gao_approximate_2004, wu_numerical_2017}, which try to find the factorizable polynomial that is closest to the input in some well defined topology. The algorithm described in \cite{wu_numerical_2017} is available in a Matlab \cite{MATLAB} package called NAClab \cite{Naclab} for numerical polynomial manipulation. In appendix \ref{sec:numerical_factorization} we provide a brief description of the problem of numerical factorization. Moreover we present an implementation that aims at finding the closest and most coarse-grained factorization of a given target polynomial.  \par
%Nevertheless, we provide a symbolic implementation of the atomic decomposition, in Mathematica, and a numerical one, using NAClab \cite{ourGithub}.
From a quantum optics perspective, this factorization implies that the state can be prepared by successively implementing each of the irreducible polynomials of the creation operators, as shown in Fig.\ref{fig:irreducible_factorization}. A recipe for implementing any such arbitrary polynomial was recently proposed \cite{andrei2025}, based also on algebraic geometry concepts. The geometric relation between all groupings of the irreducible factors of the polynomial captures all entanglement properties of the corresponding state. In the following, we develop a framework to analyze these geometrical relations and fully characterize the entanglement structure of the systems in the most atomic possible way.   

\par
\subsection{Atomic decomposition of stellar polynomials}
In this section, we use the essential variables space method presented above, together with the factorization into irreducible factors, to develop an \textit{atomic decomposition} of the stellar polynomial. This decomposition, based on the geometrical relation between the different irreducible factors of the polynomial, contains the full information about the non-Gaussian entanglement properties of the corresponding core state. Moreover, this decomposition provides a clear picture of all kinds of discrete non-Gaussian operations necessary to produce non-Gaussian entanglement. Before introducing it, we lay the groundwork by setting some definitions for the characterization of the geometrical relation between different polynomials and prove some necessary results. 

\begin{defn}
    Two multivariate polynomials $p_{1},p_{2}$ with coefficients in $\complex$ are said to be \emph{mutually disjoint} if it is possible to find a common linear change of all of their $m$ variables such that $p_{1}$ depends only on the first $m-n_1$ transformed variables and $p_{2}$ depends only on the last $n_{1}$ transformed variables, for some $n_{1} \in \mathbb{N}$. Otherwise, if such a linear change of variables does not exist, they are said to be \emph{concatenated} and we write $p_{1} \sim p_{2}$.
\end{defn}
Note that $\sim$ is reflexive ($p \sim p$ for any polynomial) and symmetric ($p_{1} \sim p_{2} \iff p_{2} \sim p_{1} $), but not transitive: in general $(p_{1} \sim p_2 )\land (p_2 \sim p_3) \nRightarrow p_1 \sim p_3$. The following result provides an efficient way to verify whether two polynomials are disjoint or not, based on the essential variables space of one of them. This is the key result in this section. 
\begin{lema}
\label{lema:orthogonalnonessential}
    Two multivariate polynomials $p_{1}$ and $p_{2}$ are mutually disjoint if and only if the essential variables of either of them are irrelevant variables with respect to the other; equivalently, $\mathbb{E}(p_{1}) \perp \mathbb{E}(p_2)$. 
\end{lema}
\begin{proof}
    The sufficiency argument is a direct consequence of the definition of disjoint polynomials: if $\forall \mathbf{v} \in \mathbb{E}(p_2)$ the variable $\mathbf{v}^T \mathbf{z}$ is irrelevant with respect to $p_1$, meaning that $\mathbf{v} \in \mathbb{E}^{\perp}(p_1)$, then $\mathbb{E}(p_{1}) \perp \mathbb{E}(p_2)$. Therefore, when $p_1$ and $p_2$ are reduced to their respective essential variables, defined by any choice of basis of  $\mathbb{E}(p_{1})$ and $\mathbb{E}(p_2)$, they do not share any variable, proving that they are disjoint. As for the necessity, suppose that $p_{1} \not\sim p_2$ and we pick the right variables so that $p_{1}$ depends upon $z_{1},...,z_{n}$ while $p_{2}$ depends upon $z_{n+1},...,z_{M}$, where no variable is shared. Then, in this basis of variables, it is already manifest that $\mathbb{E}(p_1) \perp \mathbb{E}(p_2)$, since vectors in $\mathbb{E}(p_1)$ can have only the first $n$ entries different from zero, while vectors in $\mathbb{E}(p_2)$ will start with $n$ components equal to zero. Therefore, recalling that essential spaces are covariant and their orthogonality is thus independent under unitary changes of variables, if $p_1 \not \sim p_2$, the essential variables of each one of them are irrelevant with respect to the other. 
\end{proof}

\begin{lema}
\label{lemma:groupingfactors}
    Let $\{p_1,\dots,p_n \}$ and $\{q_1,\dots, q_m \}$ be two sets of polynomials over $M$ complex variables. If $\forall j \in \{1,\dots,n\} , \ \forall k \in \{1,\dots,m\}: p_j \not \sim q_k$, then it is possible to find a unitary change of their variables such that $\{p_1,\dots,p_n \}$ and $\{q_1,\dots, q_m \}$ involve two disjoint sets of independent variables.
\end{lema}

\begin{proof}
    Using the hypotheses and Lemma \ref{lema:orthogonalnonessential}, we have that $\forall j \in \{1,\dots,n\} , \ \forall k \in \{1,\dots,m\}$ we have $\mathbb{E}(p_j) \perp \mathbb{E}(q_k)$; thus, for any fixed $j$, taking the union of the essential spaces of the $q_k$, we find $\mathbb{E}(p_j) \perp \left\{ \bigcup_{k=1}^{m}  \mathbb{E}(q_k)\right\}$. Since this holds for every $j \in \{1,\dots,n\}$, taking the union of the sets on the left-hand side we arrive at:
    \begin{equation}
        \left\{ \bigcup_{j=1}^{n} \mathbb{E}(p_j) \right\} \perp \left\{ \bigcup_{k=1}^{m}  \mathbb{E}(q_k)\right\}
    \end{equation}
    Then any orthonormal basis of $\complex^M$ that splits into an orthonormal basis of the first union and an orthonormal basis of the second union defines a change of variables in which $\{p_1,\dots,p_n \}$ do not share any variable with $\{q_1,\dots, q_m \}$.
\end{proof}

\begin{defn}
    A multivariate polynomial is said to be \emph{atomic} if it cannot be factorized into two nontrivial and mutually disjoint polynomials. In particular, irreducible polynomials are also atomic. Given a generic polynomial $p$ on $\complex$ in $M$ (essential, without loss of generality) variables, we define its \emph{atomic factorization} as:
    \[    p = \prod_{k=1}^f q^{\mathrm{at}}_{k}\]
    where each factor $q^{\mathrm{at}}_{k}$ is an atomic polynomial and they are all mutually disjoint. The \emph{atomic partition} associated with $p$ is defined as the partition of the integer $M$ induced by the number of essential variables in each factor of the atomic factorization of $p$:
    \begin{equation}
     \begin{aligned}
         & I^{\mathrm{at}} (p) \  = \ \left\{ \dim \mathbb{E}(q^{\mathrm{at}}_{1}), \dots, \dim \mathbb{E}(q^{\mathrm{at}}_{f}) \right\} , \\
       &  \sum_{k=1}^{f} \dim \mathbb{E}(q^{\mathrm{at}}_{k}) \  = \  M
         \end{aligned}   
    \end{equation}
 \end{defn}

\begin{thm}
    The atomic factorization of a multivariate, complex polynomial is unique (up to orderings of the factors) and it is possible to find a choice of variables such that each atomic factor depends on an independent set of variables.
\end{thm}
\begin{proof}
    Given a polynomial $p$ in $M$ variables, we first consider its decomposition into irreducible factors:
    \begin{equation}
        p \ = \ \prod_{j=1}^{F} p_j
    \end{equation}
    where each $p_j$ is an irreducible polynomial. Since this is unique and the factors cannot be further split, any other factorization of $p$ must come from the grouping of some of its irreducible factors into larger ones. We need to show that there is only one way to do so in such a way that the larger factors are all atomic polynomials and they are mutually disjoint. To see this, it is convenient to introduce a graph $\mathcal{G}$ whose vertex set is given by the $F$ irreducible factors of $p$ and the (undirected) edges are such that $p_j$ and $p_k$ are connected if and only if $p_j \sim p_k$ ($\mathbb{E}(p_j) \not \perp \mathbb{E}(p_k)$), in other words, if and only if there is no linear unitary change of variables under which the two polynomials involve distinct sets of the new variables. Again, since the essential spaces of the polynomials are covariant and the irreducible factorization is invariant, $\mathcal{G}$ is also invariant under different choices of the variables. It is clear that each connected component of $\mathcal{G}$ corresponds to an atomic polynomial, since all of their factors (the vertices) cannot be disjoint from the rest, given that they share some variables with other factors under all possible linear unitary transformations of the variables. Moreover, two distinct connected components represent disjoint atomic polynomials since, as a consequence of Lemma \ref{lemma:groupingfactors} and of the fact that their irreducible factors are pairwise disjoint, the essential variables of either of the two atomic polynomials are irrelevant variables for the other. The uniqueness stems from the fact that this is the only grouping of the irreducible factors leading to a factorization into atomic and mutually disjoint polynomials. 
\end{proof}

\begin{figure}[htbp]
\centering
\includegraphics[width =0.95\linewidth]{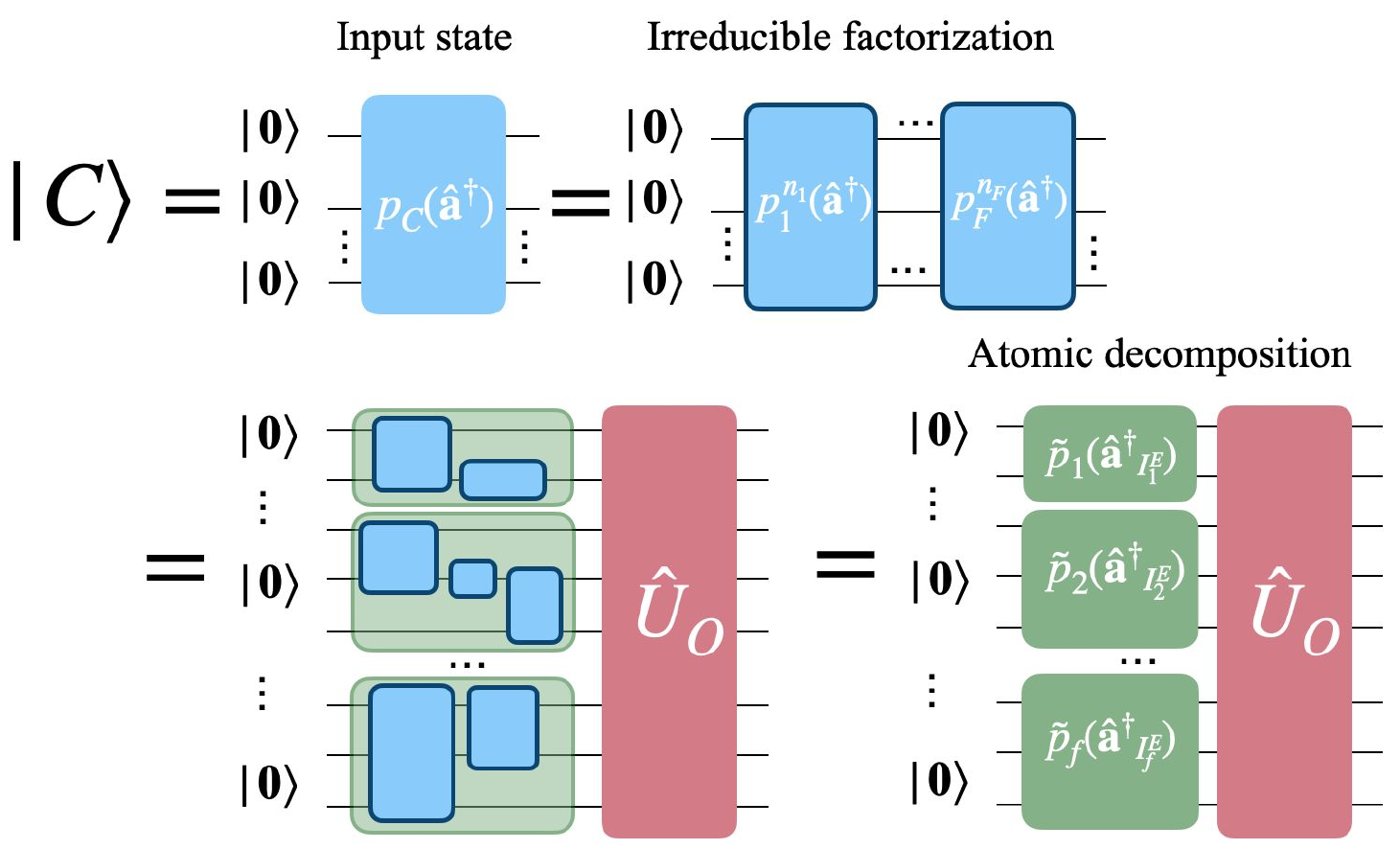}
\caption{Schematic representation of the connection between the irreducible factorization and the atomic decomposition of the stellar polynomial $p_C$, on the state preparation setting. We can naturally see that both decompositions offer complementary points of view, with the former providing a sequential decomposition of the process, and the latter providing a parallelization of the process. Moreover, the scheme to the bottom right of the figure makes clear that this decomposition contains the most fine grained information possible about the mode-intrinsic entanglement of the state.    }
\label{fig:atomic decomposition}
\end{figure}
Fig.\ref{fig:atomic decomposition} offers a pictorial representation of the connection between the factorization into irreducible polynomials and the atomic decomposition of the states. This provides a schematic representation of the discussion on the proof of the theorem and of the algorithmic procedure described below.\par
From the point of view of entanglement theory and questions on the existence of $K$-partitions, a core state associated with an \emph{atomic} stellar polynomial (reduced to its essential variables) is not passively separable in any nontrivial partition, because any factorization of its stellar polynomial will involve factors that share some essential variable. Vice versa, given the atomic partition $I^{\mathrm{at}}$ associated with the stellar polynomial of any core state, one can straightforwardly answer any question related to the existence of $K$-partitions and passive separability, which can be understood from the schematic of the atomic decomposition, provided in Fig.\ref{fig:atomic decomposition}. This partition can be derived using standard factorization algorithms to find the $F$ irreducible factors and then compute at most $F$ essential variables spaces to establish which factors belong to which connected component. In the following algorithm, we provide a procedure to determine whether a given state, described by the stellar polynomial $P$ on $M$ modes/variables, can be factorized by a mode basis transformation into some $K$-partition $I_K$.

\begin{itemize}
    \item[\bf{Step 1}] Compute the stellar polynomial $p_{C}$ associated with $\vert C \rangle$. Determine the number $E \leq M$ of essential variables, corresponding to the minimum number of independent modes that are not in the vacuum with respect to $\vert C \rangle$ and reduce $p_{C}$ to its essential variables.
    \item[\bf{Step 2}] Factor $p_{C}$ into its $F$ irreducible factors $p_1,...,p_F$, where the number $F$ will not depend on the mode / variable basis.
    \item[\bf{Step 3}] Compute $\mathbb{E}(p_1)$ and pick an orthonormal basis of $\complex^E$ that splits into a basis of $\mathbb{E}(p_1)$ and a basis of $\mathbb{E}^{\perp}(p_1)$. If, in the variables defined by this basis, the irreducible factor $p_j$ shares variables with $p_1$, then they are grouped as part of the same atomic polynomial. 
    \item[\bf{Step 4}] Once all the irreducible polynomials concatenated to $p_1$ have been found, pick the first irreducible factor $p_k$ disjoint to $p_1$ and repeat step $3$ with $p_k$ replacing $p_1$ and checking all polynomials $p_j$ with $j \geq 2, j \neq k$. 
    Proceed to the next step when no more unconnected, irreducible factors are left to test. 
    
     \item[\bf{Step 5}] Compute the atomic partition $I^{\mathrm{at}}(p_C)$ by counting the number of essential variables in each atomic polynomial obtained at the end of the previous steps. If $m = M -E > 0$, add a string of $1$ of length $m$ to $I^{\mathrm{at}}(p_C)$, representing the irrelevant modes of the initial core state. $\vert C \rangle$ is passively separable into the $K$-partition $I_{K}$ of its $M$ modes if and only if $I_{K}$ is a \emph{coarsening} of $I^{\mathrm{at}}(p_{C})$. In particular, if $K$ is larger than  $A$ (the number of atomic factors in the atomic decomposition of $p_C$) no $K$-partition with respect to which $\vert C \rangle$ is passively separable can exist.

\end{itemize}

Notice that steps 4 and 5 can be grouped together if one computes all the essential variables spaces $\mathbb{E}(p_{1}),\dots,\mathbb{E}(p_F)$ and determine the atomic decomposition from the orthogonality properties of them, but the procedure described above will spare the computation of some of these spaces by just determining the connected components instead of finding all possible connections. \par 
At any rate, the most intensive part of the algorithm is given by Step 2, the factorization of the polynomial into irreducible factors, whose complexity was discussed in Section~\ref{sec:irred_factorization}.
%which can be done in polynomial time of the degree of the polynomial, but exponential time on the number of variables.
The rest of the steps have a relatively lower complexity. Notably, the complexity of building up the structural graph increases linearly in the number of polynomial factors, multiplied by the complexity of computing the essential variables spaces. Using singular value decomposition, the complexity of computing the essential spaces is at most $\mathcal O(M^2 \binom{M+r}{r})$, where $M$ is the number of involved modes and $r$ is the degree of the polynomial. The overall complexity of finding the structural graph once the decomposition into irreducible factors has been performed is at most $\mathcal{O}(F M^2 \binom{M+r}{r})$. \par
An implementation of the method described above to obtain the atomic decomposition of a given stellar polynomial (steps 2 to 5) can be found in \cite{repo_atomic}. One version of the code is implemented in Wolfram Mathematica, and only handles polynomials defined with exact coefficients. A second version is implemented using Python and handles polynomials specified through approximate coefficients. A brief discussion of the approximate factorization is provided in Appendix \ref{sec:numerical_factorization}.

\subsubsection{Implications on the preparation of states with mode-intrinsic entanglement}
Beyond analyzing the separability with respect to a particular partition, the atomic decomposition of a core state provides a full characterization of the entanglement properties of the state and of the complexity of the state generation. As already pointed out in \cite{Kopylov_2025,andrei2025}, a given state can be generated by concatenating, in an arbitrary order, the implementation of the $F$ irreducible factors of its stellar polynomial. There are different ways in which mode-intrinsic entanglement can arise from such a structure, precisely through the generation of atomic polynomials. On the one hand, an irreducible polynomial over all the relevant variables is already an atomic polynomial. On the other hand, even with a product of linear forms, one can build an atomic polynomial, as long as the linear forms are non-orthogonal.
This kind of conditions arise when several photon additions in non-orthogonal modes are performed \cite{sperling_mode-independent_2019, Kopylov_2025}. \par
In the first case, we will speak about \emph{algebraic entanglement} and we will say that any two essential modes of an irreducible stellar polynomial are algebraically entangled, since their mode-intrinsic entanglement is inherently derived from the irreducibility property. The experimental generation of this kind of entanglement requires, in general, a rather involved state preparation procedure. An explicit construction for any such polynomial can be found in \cite{andrei2025}, as a post-selected, boson-sampling-like setup. A deterministic generation of such states would require the unitary implementation of highly nonlinear Hamiltonians \cite{Lloyd1999,arzani2025, Zhang2021}, which is in principle possible, but highly challenging.  
\par
Notice, however, that given two essential modes of an arbitrary \emph{atomic} polynomial, it might be meaningless to ask whether they are algebraically entangled or not, since they may be neither essential nor irrelevant for any single irreducible polynomial in its factorization, as we will clarify in the following examples. Similarly, it might be tempting to associate a graph to a stellar polynomial such that each vertex is an essential variable of one of its atomic factors and two vertices are linked if they are involved in the same irreducible polynomial. While this procedure will correctly identify the connected components and the atomic decomposition, the links inside each connected component will not all be basis independent: indeed, there is no way in general to choose the essential variables of an atomic polynomial in such a way that all of its irreducible factors are also reduced to their essential variables at the same time.  \par

\subsubsection{Some illustrative examples}

In what follows, we consider three examples that illustrate the working of the method, and the kind of information we can extract from it. We present examples that showcase the different ways in which mode-intrinsic entanglement can appear, and how they can be combined. The first two examples only involve linear irreducible factors, corresponding to states that can be generated by sequences of photon addition. As a final example, we show a state whose irreducible decomposition includes an irreducible factor, thus involving algebraic entanglement.  \\
\paragraph{(A)} Let us first consider the state
\begin{equation}
    |\psi_a\rangle=\frac{1}{\mathcal N}\hat a_1^\dagger\left(\frac{\hat a_1^\dagger+\hat a_2^\dagger}{\sqrt{2}}\right)\hat a_3^\dagger\left(\frac{\hat a_3^\dagger-\hat a_4^\dagger}{\sqrt{2}}\right)|0\rangle,
\end{equation}
Written in this form, the entanglement properties of the state are evident, since the state is obtained by four photon additions, the first two of which are orthogonal to the other two, but not among themselves. This implies that the state can be factorized into two sets of two modes each.  Nevertheless, this is not as evident when we expand the state and look at its stellar polynomial
\begin{equation}\label{eq:poly_state1}
    P_{|\psi_A\rangle}=\frac{1}{2 \mathcal N}\left(z_1^2z_3^2-z_1^2z_3z_4+z_1z_2z_3^2-z_1z_2z_3z_4\right).
\end{equation}
If we apply the atomic decomposition algorithm, we recover the structural graph represented in Fig.\ref{fig:exampleA_B}.
\begin{figure}[htbp]
\centering
\includegraphics[width =\linewidth]{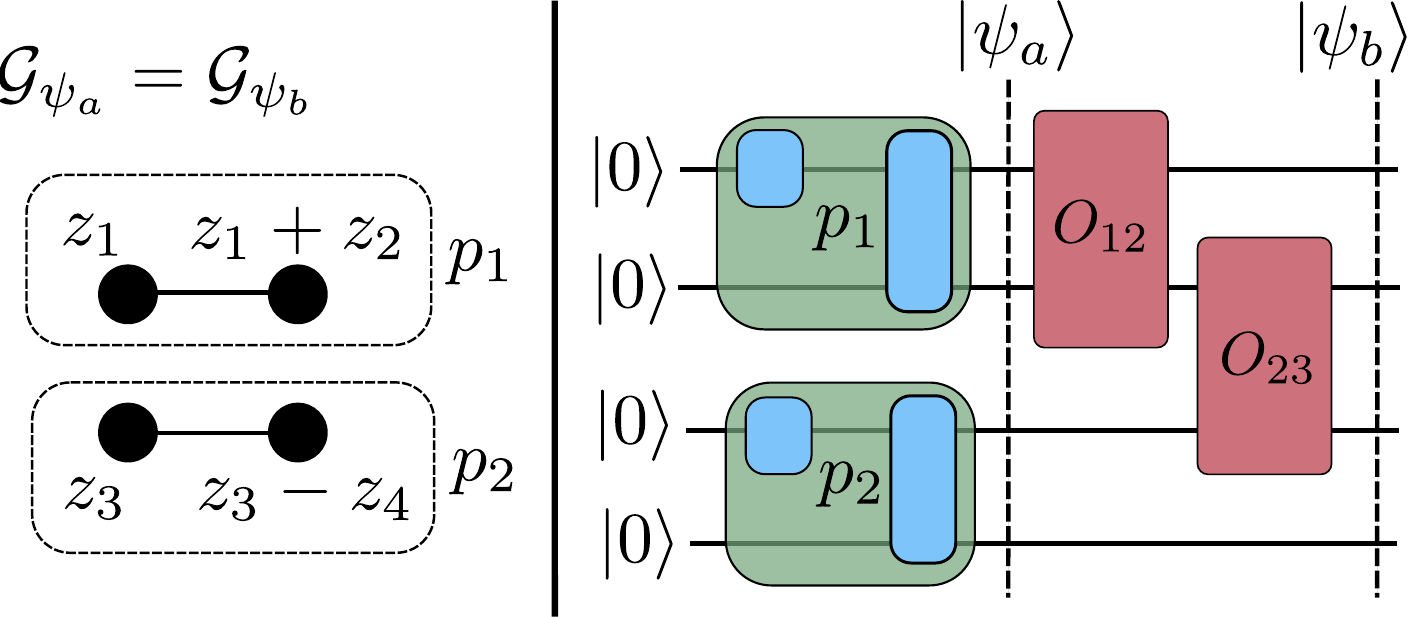}
\caption{To the left: Structural graph corresponding to the atomic decomposition of the stellar polynomials in Eq.\eqref{eq:poly_state1} and Eq.\eqref{eq:poly_scrambled_1} corresponding to the states $|\psi_a\rangle$ and $|\psi_b\rangle$ respectively. Nodes represent the irreducible factors, and edges link factors that share essential variables. The connected components of the graph correspond to the atomic factors $p_1$ and $p_2$ of the polynomials. To the right: Representation of the preparation scheme based on the atomic decomposition of the stellar polynomials. There is an ambiguity in the implementation of the inner (blue) blocks, representing the irreducible factors, as only one of them can be brought to its essential variables.}
\label{fig:exampleA_B}
\end{figure}

This structural graph has two connected components, each of which is an atomic polynomial: $p_1(z_1,z_2)=z_1 (z_1+z_2)$ and $p_2(z_3,z_4)=z_3 (z_3-z_4)$. Each of the atomic polynomials creates mode-intrinsic entanglement within the modes it acts on. On the other hand, the two sets of two modes are disentangled with respect to each other. This example was already written in the basis in which its atomic decomposition could be found, so we just present it for the sake of making concepts clear.  \\

\paragraph{(B)} The same state scrambled to a different mode basis.\par 
We analyze the same state as above but scrambled in an arbitrary way over the modes it populates. We apply a rotation given by a beam splitter mixing the first two modes, followed by a similar beam splitter interfering modes 2 and 3. Both beam splitters are set with parameter $\theta=3/10$. The new state is 
\begin{equation}\label{eq:poly_state2}
    |\psi_b\rangle=\hat U_{O_{23}}(\theta)\hat U_{O_{12}}(\theta)|\psi_a\rangle,
\end{equation}
whose stellar polynomial is 
\begin{equation}\label{eq:poly_scrambled_1}
\begin{split}
    & p_{|\psi_b\rangle}(\vec z)= 0.088 z_1^4-0.033 z_1^3 z_2+0.54 z_1^3 z_3\\&-0.311 z_1^3 z_4
     -0.0202 z_1^2 z_2^2-0.015 z_1^2 z_2 z_3\\
    &+0.68 z_1^2 z_3^2-0.88 z_1^2 z_3 z_4+ 0.0115 z_1 z_2^3\\ 
    & -0.147 z_1 z_2^2 z_3+0.077 z_1 z_2^2 z_4+0.29 z_1 z_2 z_3^2\\
    &-0.28 z_1 z_2 z_3 z_4-0.52 z_1 z_3^3+0.57 z_1 z_3^2 z_4\\
    &-0.00145 z_2^4 +0.031 z_2^3 z_3-0.0166 z_2^3 z_4\\
    &-0.153 z_2^2 z_3^2+0.171 z_2^2 z_3 z_4-0.120 z_2 z_3^3\\
    &+0.118 z_2 z_3^2 z_4+0.080 z_3^4-0.083 z_3^3 z_4\\
    &+0.0192 z_1^2 z_2 z_4.
    \end{split}
\end{equation}
The underlying structure of this polynomial is much less apparent than the one of the state before applying the mode basis transformation. The polynomial considered for the implementation of the algorithm using symbolic computation tools must be the exact one, which keeps, for example $\cos(3/10)$ unevaluated. On the other hand, we can employ the method described in Appendix \ref{sec:numerical_factorization} to perform a numerical factorization and with it find the atomic decomposition. In the case of this, still relatively simple polynomial, it is possible to perform both a numerical and a symbolic analysis. 
\par
When we input the polynomial \eqref{eq:poly_scrambled_1}, to the atomic decomposition algorithm, it correctly identifies the four irreducible components, which are then grouped into a graph consisting of two independent connected components. 
\begin{equation*}
    \begin{aligned}
         q_1(\vec z) =&
        \frac{1}{2} \left(z_1+z_1 \cos \left(\frac{3}{5}\right)-z_2 \sin \left(\frac{3}{5}\right)\right. \\
        & \left. -2 z_3 \sin \left(\frac{3}{10}\right)\right)\\
         q_2(\vec z) =& \frac{1}{2} \left( z_1 \left(1+2 \sin \left(\frac{3}{10}\right)+\cos \left(\frac{3}{5}\right) \right)\right.\\&\left.-z_2 \sin \left(\frac{3}{5}\right)+ 2 z_2 \cos \left( \frac{3}{10} \right) \right. \\
         & \left.   - 2 z_3 \sin \left( \frac{3}{10} \right) \right)\\
         q_3(\vec z) =& \frac{1}{2} \left(z_1 \sin \left(\frac{3}{5}\right)+z_2 \left(\cos \left(\frac{3}{5}\right)-1\right)\right.\\
         &\left. +2 z_3 \cos \left(\frac{3}{10}\right)\right)\\
         q_4(\vec z) =& \frac{1}{2} \left(z_1 \sin \left(\frac{3}{5}\right)+z_2 \left(\cos \left(\frac{3}{5}\right)-1\right) \right. \\
         & \left. +2 z_3 \cos \left(\frac{3}{10}\right)-2 z_4\right)
    \end{aligned}
\end{equation*}
which are grouped into a structural graph $\mathcal G_{|\psi_2\rangle}=\{\{q_1,q_2\},\{q_3,q_4\}\}$. The transformation that brings the state to its atomic structure was obtained to be 
\begin{equation*}
    O=\left(
\begin{array}{cccc}
 0.912668 & 0.29552 & 0.282321 & 0 \\
 -0.282321 & 0.955336 & -0.0873322 & 0 \\
 -0.29552 & 0 & 0.955336 & 0 \\
 0 & 0 & 0 & 1, \\
\end{array}
\right)
\end{equation*}
which corresponds, as expected, to the inverse of the transformation applied on the input state. The structural graph, with the irreducible factors expressed in the \textit{atomic basis} is given by the same represented in Fig.\ref{fig:exampleA_B}. This highlights the utility of the atomic decomposition, as an invariant under mode basis transformations, to determine whether two states can be related by passive linear optics, which in itself is an outstanding problem \cite{Migdal_2014,parellada_2023}. Though the structural graph is invariant under transformation, it is important to keep in mind that there is an ambiguity in the choice of mode basis bringing the state to its atomic form. If we had changed the order of the irreducible factors within a single atomic component we would end up with a different choice of basis. In the right block of Fig.\ref{fig:exampleA_B} we represent the scheme for the preparations of states $|\psi_a\rangle$ and $|\psi_b\rangle$ based on the atomic decomposition. The ambiguity on the choice of basis reflects on the different possible choices of inner (blue) blocks, representing the implementation of the irreducible factors within an atomic component.

\paragraph{(C)} As mentioned previously in the discussion, another family of not passively separable core states is provided by those whose stellar polynomial is in itself irreducible. We consider now an example in which such a behavior is observed: 
\begin{equation}
    |\psi_c\rangle=\frac{1}{\mathcal N}\frac{\hat a_1^2+\hat a_2^2+\hat a_3^2}{\sqrt{3}} \frac{\hat a_3^2-\hat a_4^2}{\sqrt{2}}|0\rangle.
\end{equation}
The corresponding stellar polynomial is given by 
\begin{equation}\label{eq:poly_non_linear_factors}
    P_{\psi_c}(\mathbf z)=\frac{1}{\mathcal N}\frac{\left(z_1^2+z_2^2+z_3^2\right)\left(z_3^2-z_4^2\right)}{\sqrt{6}},
\end{equation}
where the first factor is irreducible and the second can be factorized into two linear terms. When we perform the atomic decomposition we obtain that the polynomial has a single connected component, thus a single atomic factor, which in turn means that there is no partition of the four involved modes, for which the state is passively separable. The corresponding atomic decomposition is represented in Fig.\ref{fig:example_C} through its structural graph.\par 
\begin{figure}[htbp]
\centering
\includegraphics[width =0.9\linewidth]{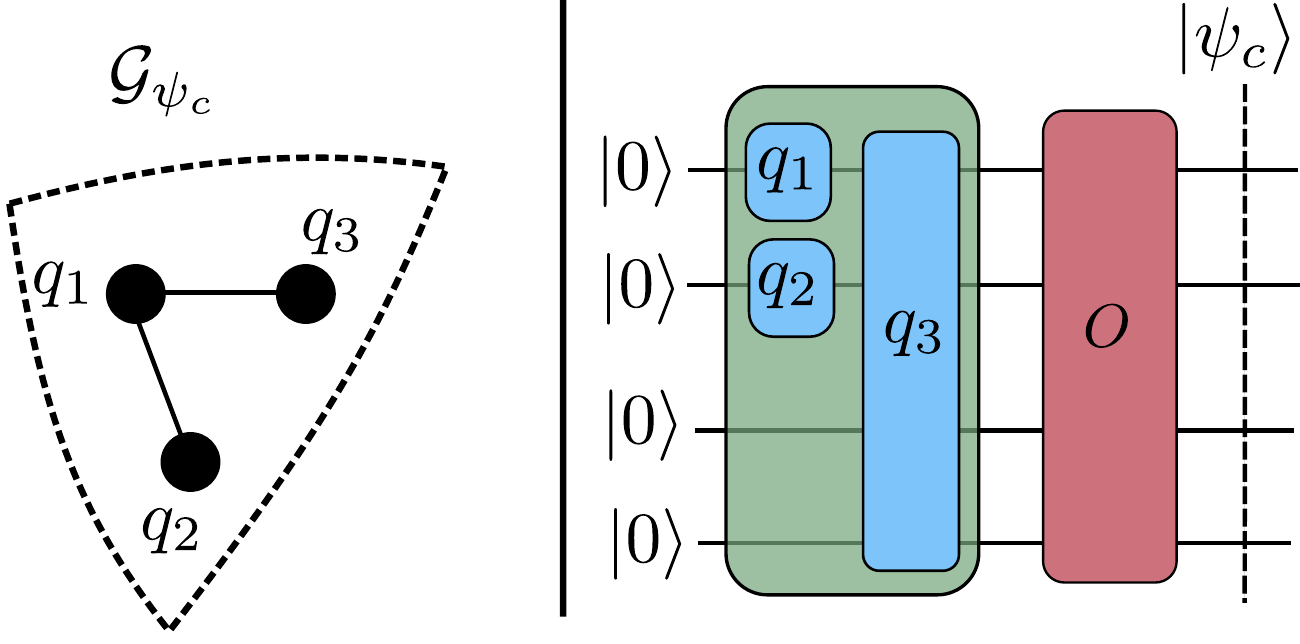}
\caption{To the left: Structural graph corresponding to the atomic decomposition of the stellar polynomial in equation \eqref{eq:poly_non_linear_factors}. To the right: Representation of the preparation of the state $|\psi_c\rangle$ based on the atomic decomposition of its stellar polynomial. One choice of many possible was made considering the implementation of the irreducible factors $q_1, q_2, q_3$. The transformation $O$ involving only modes inside the atomic component reflects this ambiguity.}
\label{fig:example_C}
\end{figure}
This is an appropriate example to reflect on the ambiguity on deciding whether two modes will be algebraically mode-intrinsic entangled or not. The linear factors $q_1=z_3-z_4$ and $q_2=z_3+z_4$ do not share essential variables. Nevertheless, they both share essential variables with the non-linear factor $q_3=z_1^2+z_2^2+z_3^2$. As a consequence, there is no way to separate them. \par
\par 

The examples considered in this section are all very simple, with the objective of making the exposition clear and light. In \cite{repo_atomic} a set of more complicated examples are also considered. There, we mainly focus on the numerical procedure described in Appendix \ref{sec:numerical_factorization}, which is more appropriate for considering arbitrary polynomials. We test in particular polynomials generated under random beam splitter networks applied over $M$ modes, with the atomic partition chosen randomly prior to the interferometer. There is a trade-off between the number of modes and the degree of the polynomials that we can consider. We have been able to consider examples with up to 8 variables and degree 4, and 3 variables with degree 12, for example.   \par
In what follows we describe two special cases for which the analysis of the passive separability of core states can be analyzed using different, specialized, procedures. 

\section{Passive separability in special cases}
\label{sec:special_cases}

\subsection{Passive separability of two-mode core states} 
Two-mode core states of an arbitrary stellar rank $r$, provide a relevant and enlightening particular case of the analysis we have considered so far. The structural graph of two-mode core states can either have a single connected component or two connected components. Passively separable core states are those with two connected components, each of which involves a single mode, or equivalently a single variable stellar polynomial. Univariate polynomials all have a discrete number of zeros, equal to their stellar rank (degree of the polynomial) \cite{chabaud_holomorphic_2022}. When elevated to the full two-mode phase space, these discrete zeros, each extend as a plane. Planes corresponding to different modes will be orthogonal to each other. This orthogonality is preserved by passive operations, and is thus the signature of passive separability. The orthogonality between these planes is equivalent to the disjointness of their corresponding linear factors. In what follows we formalize this connection, to provide a tailored algorithmic procedure for the analysis of passive separability of two-mode core states. A formalization of these ideas is provided in the Appendix \ref{sec:Appendix_A}, leading to the following theorem.\par

 \begin{thm}\label{theorem:orthogonality_is_equivalent_to_passive_separability}
    A pure core state $|C\rangle$, of a two-mode bosonic system, is passively separable if and only if the hypersurfaces of zeros of its stellar polynomial $p_\star (\mathbf z)$, $V(p_\star)$ can be separated into two sets of hyperplanes orthogonal to each other, \ie 
    \begin{equation*}
        V(p_\star (\mathbf z))=V_1 \cup V_2, \text{ such that } V_1 \perp V_2.
    \end{equation*}
    \begin{proof}
        The proof of this theorem is a direct application of all the results listed in Appendix \ref{sec:Appendix_A}.
    \end{proof}
 \end{thm}

The approach described in this section is very similar to the ones employed in \cite{sperling_mode-independent_2019} and \cite{Migdal_2014}. In both cases they consider homogeneous polynomials of two variables, for which the factorization is always into linear terms. Here we consider also non-homogeneous polynomials, for which there is the possibility of nonlinear irreducible factors, which immediately implies some mode-intrinsic entanglement. This is the case, for example, for maximally entangled states of the form $|\psi\rangle\propto \sum_{n=0}^{N}|n,n\rangle$. \par 
In many cases of interest it is possible to directly obtain, in an efficient way, a parametrization of all the zero surfaces and thus directly check the conditions above. Nevertheless, for general high stellar rank states obtaining such parametrizations may not always be simple. In Appendix \ref{sec:Appendix_A} we outline a simple procedure to check the conditions above, without having direct access to a parametrization of the zero surfaces. In what follows we study some relevant two-mode core states as an example of the applicability of the ideas we have discussed. 
\subsubsection{States with stellar rank 1}
Let us consider the first non-trivial case, which is given by all possible two-mode core states of stellar rank 1. They are given by 
\begin{equation}
    |C_{1}\>=c_0|00\>+c_{1}|01\>+c_{2}|10\>,
\end{equation}
where $c_{i}\in \complex$, and $\sum |c_i|^2=1$. The corresponding stellar polynomial is
\begin{equation}
    p_{|C_{1}\>}(z_1,z_2)=c_0+c_{1}z_2+c_2z_1.
\end{equation}
The zeros of this polynomial are located in the plane 
\begin{equation}
    z_{2}=-\frac{c_{2}}{c_1}z_1-\frac{c_0}{c_1}.
\end{equation}
The fact that there is only one plane of zeros implies that the state is passively separable, as we can bring this plane to be parallel to any of the complex planes $z_1 $ or $z_2$, \ie we can localize the photon into any of the two modes. We can then conclude that any core state of stellar rank one can be passively disentangled, and in consequence any state of stellar rank one is Gaussian separable. The result for core states was already pointed out in \cite{sperling_mode-independent_2019}. A similar result was described in \cite{barral_metrological_2024} for the case of single photon subtracted states, which span all possible states of stellar rank one. 

\subsubsection{NOON states}
We consider now the example of NOON core states
\begin{equation}
    |NOON\>=\frac{|N0\>+|0N\>}{\sqrt{2}}.
\end{equation}
The stellar function is given by 
\begin{equation}
    P_{\text{NOON}}(z_1,z_2)=\frac{1}{\sqrt{2(N!)}}\left(z_{1}^{N}+z_{2}^{N}\right).
\end{equation}
The zeros of this polynomial are given by planes of the form
\begin{equation*}
    z_2= \kappa_j z_1,
\end{equation*}
$\kappa_j=e^{i \pi \frac{2 j+1}{N}}$. For $N=2$, this corresponds to two planes which are  orthogonal to each other, which is consistent with the fact that this state is equivalent to the Hong-Ou-Mandel state (up to a local rotation - required for disentangling the state). On the other hand, for $N\geq 3$, the planes are no longer separable into two orthogonal sets. To see this it suffices to take two consecutive planes, for $N=3$, parametrized $v_1=(z_1,e^{i\pi/3}z_1)$ and $v_2=(z_1,e^{i\pi}z_1)$. The inner product between them is $v_1^*\cdot v_2=|z_1|^2(1+e^{-2i\pi/3})\neq 0$. For this to be zero we require $\kappa_2^* \kappa_1=-1$.\par 
Interestingly, the factorization of the NOON states for $N\geq 3$ into non-orthogonal linear terms, implies that the state can be produced by a sequence of photon additions on non-orthogonal modes. The entanglement is thus not algebraic, as the corresponding atomic polynomial is not irreducible. \par

 In the Appendix \ref{sec:appendix_B} we apply this framework to the analysis of other two-mode states of relevance like two photon subtracted states and maximally entangled core states of arbitrary stellar rank. In what follows we analyze multimode states, for which new techniques are introduced.

\subsection{Passive separability of states of stellar rank 2}
States of stellar rank 2 over an arbitrary number of modes provide a further special case which can be handled using a tailored procedure. The structural graph of these states contains either a single node (a single irreducible polynomial of degree 2) or two nodes (two linear factors). The latter case implies that the state is composed of two photons populating two, possibly non-orthogonal, modes. If these modes are orthogonal the state is separable with respect to any bipartition. If they are not orthogonal, the state is still separable with respect to any partition in which some subset includes more than two modes. The case of a single node (a single irreducible polynomial) simplifies to analyzing the essential variables of this polynomial with respect to the partition of interest. In what follows we formalize these statements and include a specific method to analyze these states. Importantly, the structural graph contains at most two nodes, implying that only bipartite entanglement is relevant.    \par 

Consider an $M$-modes (pure) core state of stellar rank $2$. It will be described by a second order polynomial in $\mathbf{z} \in \complex^{M}$ of the type:
\begin{equation}
\label{e:genpol2}
\mathbf{z}^{T} \mathbf{{A}} 
\mathbf{z} \ + \ \bm{l}^T\mathbf{z} \ + \ c_{0}
\end{equation}
where $\mathbf{A}$ is an $M \times M$ symmetric \emph{complex} matrix, $\bm{l} \in \complex^{M}$ and $c_{0} \in \complex$. If $M \geq 3$ and we are asking if the state can be (passively) factorized into a state of $M_{1}$ modes and a state of $M_{2}$ modes with $M_{1} \leq M_{2}$ and $M = M_{1} + M_{2}$, we would like to bring the polynomial to the form:
\begin{equation}
\label{eq:genfactoriz2}
\left( \bm{x}^{T}  \mathbf{K}_{2} \bm{x}  +  \bm{k}_{1}^T \bm{x} + k_{0} \right) \left( \bm{y}^{T} \mathbf{R}_{2} 
\bm{y} + \bm{r}_{1}^T \bm{y} + r_{0} \right)
\end{equation}
where $\bm{x} \in \complex^{M_{1}}$ and $\bm{y} \in \complex^{M_{2}}$, such that $\mathbf{U} \ ( \bm{x}, \bm{y} ) = \mathbf{z}$.
Since the degree must be preserved, either $\mathbf{K}_{2} =  \mathbf{R}_{2} = 0$ or one of the two factors must be constant. Let us consider the first case first. Changing basis in $\complex^{M_{1}}$ and in $\complex^{M_{2}}$, we can rewrite the factorized polynomial in Eq.(\ref{eq:genfactoriz2}) as:
\begin{equation}
\label{eq:factorlin}
(k_{1} x + k_{0} ) ( r_{1} y + r_{0} ) 
\end{equation}
But now if $M_{2} \geq 2$, which is always the case if $M \geq 3$, we can consider both modes $x$ and $y$ as part of the $M_{2}$ modes in one of the factors. This means that the singular value decomposition of $\mathbf{A}$ (which, thanks to the symmetry of $\mathbf{A}$, is of the form $\mathbf{V}^{T} \mathbf{D} \mathbf{V}$ with $\mathbf{V}$ unitary and $\mathbf{D}$ diagonal and complex) will have only $2$ nonzero singular values or, equivalently, the initial polynomial will involve only 2 essential variables. Therefore, for $M \geq 3$ the case in which only one of the two factors in Eq.(\ref{eq:genfactoriz2}) is nonconstant is the most general one. In that case, at least $M_{1}$ singular values of $\mathbf{A}$ must be zero and, in the singular-value basis of $\mathbf{A}$, the vector $\bm{l}$ must have support on a space of dimension $\leq M_{2}$ whose orthogonal complement is completely contained in $\ker ( \mathbf{A} )$. 
The only case left for stellar rank $2$ is $M =2$. In this case the question is whether we can bring it to factorization of the two modes, \ie $M_1=1,M_2=1$. Therefore, Eq.(\ref{eq:factorlin}) no longer provides a straightforward factorization. In this case, $\mathbf{A}$ can have $2$ nonzero singular values, but they must be equal to each other in absolute value, but with opposite signs. If this condition is satisfied, then the core state is factorizable in a $1+1$ modes decomposition if either $\bm{l}$ is not an eigenstate of $\mathbf{A}$ and $c_{0} \neq 0$, or $\bm{l}$ is an eigenstate of $\mathbf{A}$ and $c_{0} = 0$.

\begin{thm}
Let the second-order polynomial:
\begin{equation}
\mathbf{z}^{T} \mathbf{{A}} 
  \mathbf{z} \ + \ \bm{l}^T \mathbf{z} \ + \ c_{0}
\end{equation}
describe an $M$-modes core state of stellar rank $2$. Then the state is passively factorizable as $\vert \psi_{1} \rangle \otimes \vert \psi_{2} \rangle$, where $\vert \psi_{1} \rangle$ describes $M_{1}$ modes and $\vert \psi_{2} \rangle$ describes $M_{2}$ modes such that $M_{1} \leq M_{2}$ and $M = M_{1} + M_{2}$ if and only if at least one of the following conditions holds:

\begin{itemize}
\item If $M_{2} > 1$, $\mathbf{A}$ must have rank $\leq M_{2}$ and, in any basis where the nonzero singular values of $\mathbf{A}$ are put first, $\bm{l}$ has at most the first $M_{2}$ entries different from zero. 
Equivalently, if the polynomial only involves $\leq M_{2}$ essential variables. 
\item If $M_{2} = 1$, then $M_{1} = M_{2} = 1$ and $M=2$. In this case, either the second-order polynomial in 2 variables must factorize as a product of two linear factors, each involving a different essential variable, or the whole polynomial must involve a single essential variable. 
\end{itemize}
\end{thm}
\section{Beyond pure states}
The analysis we have presented so far has been restricted to pure states. For those, there is no ambiguity in the definition of passive or Gaussian separability. On the other hand, for mixed states a choice must be made as to how to define those concepts. 
There are two options: we could consider a state to be passively (Gaussian) separable if it can be decomposed into a mixture of passive (Gaussian) separable states, \ie 
\begin{equation}\label{eq:definition_passive(Gaussian)_separability_convex}
    \hat \rho=\sum_{\gamma} p(\gamma) \hat U_{\gamma} \hat \rho_{I}^{(\gamma)}\otimes \rho_{J}^{(\gamma)}\hat U^{\dagger}_{\gamma};
\end{equation}
or we could consider it to be passively (Gaussian) separable, if there is a single operation $\hat U$ that brings the state to a separable form, \ie 
\begin{equation}\label{eq:non-convex_separability}
    \hat U \hat{\rho}\hat U^{\dagger}=\sum_{\gamma} p(\gamma) \hat \rho_I^{(\gamma)} \otimes \hat \rho_J^{(\gamma)}.
\end{equation}
The choice of one definition or another has profound implications for the structure of the set of passively (Gaussian) separable states. If we consider mode-intrinsic and genuine-non-Gaussian entanglement as quantum resources \cite{Chitambar_Gour_2019}, these are the sets of \textit{free states} of the corresponding theories. A definition according to Eq.\eqref{eq:non-convex_separability} implies that mixtures of passively or Gaussian separable states may be resourceful. As a result, the corresponding set of free states is non-convex. On the other hand, a definition according to Eq.\eqref{eq:definition_passive(Gaussian)_separability_convex} leads to a convex resource theory.\par 
This kind of ambiguity already appears in the construction of a resource theory of non-Gaussianity \cite{Albarelli2018}. Depending on the application of interest, either a convex or a non-convex resource theory may make more sense. While in some cases, having entanglement in any mode basis may be enough to consider a state as resourceful (corresponding to the non-convex resource theory), in others this may not be the case. In \cite{chabaud_resources_2023} it is argued that passive separability renders a certain class of sampling problems efficiently simulatable with classical resources. This argument is independent of which definition of passive separability we use, as long as we know the decomposition. Therefore, in this context, a convex resource theory is required. \par 
The methods to detect mode-intrinsic or genuine-non-Gaussian entanglement must be tailored to the specific choice of definition. In \cite{lopetegui_detection_2024}, a method was proposed to witness mode-intrinsic entanglement, considering Eq.\eqref{eq:non-convex_separability} as the definition of passively separable states. On the other hand, considering a convex resource theory permits to deploy all the machinery available for convex resource theories, from resource measures to resource witnesses. In \cite{Lachman2026}, a method was proposed to detect both mode-intrinsic and genuine-non-Gaussian entanglement within a convex resource theory. 
 The method presented in \cite{Lachman2026} exploits the fact that for every mode-intrinsic (genuine-non-Gaussian) entangled state, there is a ball around it, in the topology induced by the trace distance for example, where only other resourceful states are found. Therefore, if we evaluate the fidelity to a target pure resourceful state, if it is above a certain threshold, we can certify the presence of mode-intrinsic or genuine-non-Gaussian entanglement. The convexity of the set of free states significantly simplifies the computation of this threshold by restricting the optimization to pure states. While in \cite{Lachman2026} only the case of two-mode states was considered, this framework can be extended to consider multimode states and arbitrary partitions. This offers a viable route to witness whether a given state can be constructed by mixtures of structures of the form described in Fig.\ref{fig:atomic decomposition}, for a given partition $I_K$. Iterating over the different partitions of the $m$ modes may allow one to determine the coarse-graining of the modes that best corresponds to the atomic decomposition for pure states. On the other hand, the method employed in \cite{Lachman2026} to compute the thresholds efficiently cannot be directly extended to multimode states. Nevertheless, different methods may be considered in the future to speed up this computation in different regimes. \par 

\section{Conclusions and outlook}
In this work we have put together and analyzed the different ways in which entanglement can appear in non-Gaussian quantum states. Intrinsic forms of non-Gaussian entanglement include mode-intrinsic entanglement and genuine-non-Gaussian entanglement. Based on the stellar representation, which associates a multivariate complex polynomial to every pure core state of a bosonic system, we have developed methods to analyze the entanglement properties of those states based on the analysis of algebraic geometric properties of their associated polynomials. Notably, we propose a method to obtain an \textit{atomic decomposition} of the stellar polynomial of a quantum state, that fully reflects the inherent entanglement structure of the state of interest.  
%This atomic decomposition has a clear analogous graph structure, which we call structural graph, and that fully reflects the inherent entanglement structure of the state of interest. 
\par 
Moreover, the atomic decomposition of the polynomial provides key information about the state preparation tools required to produce the state of interest, narrowing the application of non-Gaussian processes as much as possible. This decomposition also helps us identify two  key ways for generating mode-intrinsic entanglement. Of these, the one with the most immediate applicability is interleaving single-photon additions on non-orthogonal modes. On the other hand, the implementation of irreducible polynomials of creation operators over several variables has a similar effect, producing a form of entanglement that we chose to refer to as \textit{algebraic mode-intrinsic entanglement} (AMIE). Further analysis of the interest of one procedure or another is left for future works. \par
This work, along with other recent publications \cite{Kopylov_2025,andrei2025}, builds on the connection between algebraic geometry tools and quantum information. We believe that this may open a fruitful cross-fertilization between the two domains, that may bear unexpected results. 
\par
The framework we have presented is restricted to states with a finite stellar rank, whose core state can be described with a stellar polynomial of finite degree. Nevertheless, this set of states is dense in the state space \cite{Chabaud_Roeland_2021}. Any state, even those with infinite stellar rank, can be arbitrarily well approximated with a finite stellar rank state. To apply our methods in states of infinite stellar rank we can set a precision $\epsilon$ in the approximation of the state through finite stellar rank states. The minimum stellar rank required to be $\epsilon$-close in fidelity to $|C\rangle$ is $r_{|C\rangle,\epsilon}$, the $\epsilon$-approximated stellar rank of the state \cite{Hahn2026}.

\begin{acknowledgments}
 We thank Andrei Aralov, Cisco Gooding, Mathieu Isoard, and Ulysse Chabaud for fruitful discussions. This work was financially supported by the ANR JCJC project NoRdiC (ANR-21-CE47-0005), the Plan France 2030 through the project OQuLus (ANR-22-PETQ-0013), and the HORIZON-EIC-2022-PATHFINDERCHALLENGES-01 programme under Grant Agreement Number 101114899 (Veriqub).
\end{acknowledgments}
\section*{Author contribution statement}
All authors contributed equally to this work. Large language models were used for bibliographic search, grammar proofreading, and code-writing. 

\bibliographystyle{quantum}
\bibliography{biblio}

\onecolumn
\appendix
\section{Numerical factorization}\label{sec:numerical_factorization}
In this section, we discuss the problem of numerical factorization, which we employ as a subroutine to obtain the atomic decomposition of a given complex polynomial. In the main text, we discuss the fact that exact factorization of polynomials is a numerically unstable problem \cite{gao_approximate_2004,wu_numerical_2017}: a small perturbation of the coefficients of a factorizable polynomial renders it irreducible with unit probability, since the subspace of factorizabl polynomials has measure zero and it is nowhere dense in the space of all polynomials. A common approach to this problem consists of trying to find the closest factorized polynomial to a target, provided a metric in the space of polynomials and a distance threshold, on which we will elaborate later. Nevertheless, as was shown in \cite{wu_numerical_2017}, this may lead to difficulties. The authors constructed a factorized example for which the closest factorized polynomials were not given by the underlying factorization, but rather factorizations into higher order irreducible factors. A detailed discussion of how to circumvent this problem and rigorously define numerical polynomial factorization was provided in \cite{wu_numerical_2017}. Covering in full detail that discussion is beyond the scope of the present work. Instead, in what follows we briefly illustrate how the problem of assigning an atomic partition to a purported (pure) core state within a given fidelity threshold can be addressed.  
\par 
As we have proven in the main text: the existence of a unique irreducible factorization for polynomials with coefficients defined exactly guarantees the existence of a unique atomic decomposition (up to local changes of variables within separated atomic factors). If instead the polynomial is known with some degree of uncertainty, we need to modify the specification of the task. The uncertainty in the polynomial is equivalent to an uncertainty in the underlying quantum state. Geometrically, this means that we only know an open region in the space of states where our target lies.\par 
We first need to define a metric in the space of polynomials that directly relates to a similar notion in the space of states. Consider the following mapping between the coefficients of a polynomial and the coefficients of Fock states in the Fock expansion of the corresponding state
\begin{equation}
    p^{(\mathbf k)} \sqrt{\mathbf n!}=C_{\mathbf n}.
\end{equation}
Then the Bombieri norm of the polynomial is defined as
\begin{equation}
    \|p\|_B=\sum_{\mathbf k} |p^{(\mathbf k)}|^2 \mathbf n!, 
\end{equation} 
which for stellar polynomials matches the square-norm of the corresponding core state, and therefore equals $1$. This norm can be seen as derived from the inner product $\langle p_1,p_2\rangle_B=\sum_{\mathbf k} p_1^{(\mathbf k)*} p_2^{(\mathbf k)} \mathbf n!$, equivalently to the overlap between quantum states. Finally, it induces a notion of distance through:
\begin{equation}
    d_B(p_1,p_2)=\|p_1-p_2\|_B,
\end{equation}
which evaluates to 
\begin{equation}
    d_B(p_1,p_2)=2\left(1-\text{Re}(\langle p_1,p_2\rangle_B)\right),
\end{equation}
for stellar polynomials. This is equivalent to the Euclidean distance between two quantum states, given by 
\begin{equation}
    ||\psi_1\rangle-|\psi_2\rangle|^2=2(1-\sqrt{F}e^{i\phi}),
\end{equation}
where $F=|\langle \psi_1|\psi_2\rangle|^2$ is the fidelity between the two states, and $\phi=\arg(\langle \psi_1|\psi_2\rangle)$. Notice that the Euclidean distance is sensitive to a global phase. Nevertheless, it is convenient for the specification of an optimization, which would in any case modulate the phase to minimize the distance, equivalently maximizing the fidelity.

Within a ball of radius $\delta$ in the Bombieri norm around our approximate target state/polynomial, there will typically be many different polynomials with an exact factorization. Each of these can be considered an approximate factorization of our target. However, given two admissible candidate factorizations, at distance $d_1$ and $d_2$ from our approximate state and having the same irreducible partition (i.e., the same partition of the total degree into degrees of each factor), there is no way to argue that the closest one will lead to the finer \emph{atomic} partition \emph{of modes}. In other words, if we are after the "optimal" atomic partition (in a sense that will be made precise in the following), it may be inconvenient to start from the irreducible decompositions and more direct numerical optimization methods to find the closest states with nontrivial atomic partitions to our target should be deployed.

For partitions in more than two sets, it is not possible to establish a meaningful total ordering. On the one hand, whenever a partition is a coarsening of another, we want to keep the finer one, but some partitions can be incomparable (e.g. $(4,1,1)$ and $(2,2,2)$), leading to ambiguities in the choice. As a matter of fact, it turns out that only partial orderings are admissible \cite{wu_numerical_2017}. The natural partial ordering, which captures the idea of coarsening, can be represented as a directed graph. In Fig.\ref{fig:directed_graph_poset_degree_6}, we can observe the graph for the partition of the integer $6$, where the ordering relations discussed above can be observed. \par
\begin{figure}[htbp]
    \centering
    \includegraphics[width=0.25\linewidth]{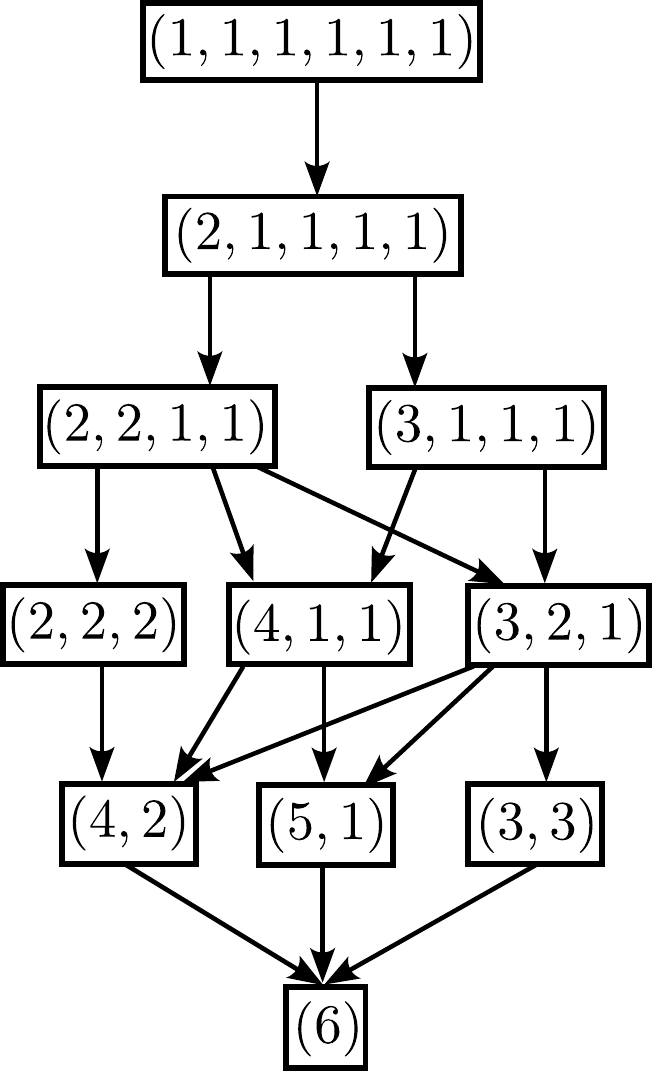}
    \caption{Directed graph representing the partial ordering of all the partitions $P_d$ of an integer $d=6$. The impossibility to reach $(2,2,2)$ from $(3,1,1,1)$ reflects the fact that they are incomparable. On the other hand if a node $a$ can be reached from a node $b$, we say that $a<b$.  }
    \label{fig:directed_graph_poset_degree_6}
\end{figure}
By fixing the number of modes $M$ and constructing the corresponding graph of partially ordered partitions, we can define the approximate atomic partition describing our target state within fidelity $\delta$ in the following way:
\begin{defn}
    \textbf{Inexact atomic decomposition: }Given a pure core state $|C\rangle$ of $M$ modes, known only with some uncertainty with respect to the actual exact target state $|C_t\rangle$, \ie there is a $\delta\geq 0$, so that $|\langle \psi_t|\psi\rangle|^2\geq 1-\delta$. This implies that the corresponding stellar polynomial $p_{C}(\mathbf z)$, is close in the Euclidean distance to the exact target polynomial $|\langle p_{C_t},p_C\rangle|^2\geq 1-\delta$. We formally define the approximate atomic partition of $p_C$ as the most-fine grained atomic partition of the $M$ modes occurring within a distance $\delta$ from the target. In the occurrence of multiple inequivalent partitions among the finest ones, the partition associated with the state which is closest to the target must be picked. 
\end{defn}

If, instead of constructing the graph starting from the number of modes $M$, we start with the total number of photons $r$ (equivalently, the degree of the stellar polynomial), we can rank the \emph{irreducible} factorizations occurring within a distance $\delta$ from the target and define a notion of \emph{approximate irreducible factorization} within a distance $\delta$. Importantly, according to the remarks above, this approximate irreducible factorization will in general not correspond to the approximate atomic partition. Nevertheless, the approximate irreducible factorization can be useful for other scopes, e.g. to determine the most practical approximate preparation procedure of the target using linear optics and photons heralding.

\section{The particular case of two modes}\label{sec:Appendix_A}
Any two-mode, passively separable core state can be written as 
\begin{equation}
    |C\rangle= \hat{U}_{\mathbf{O}} |\psi_1\rangle \otimes |\psi_2\rangle.
\end{equation}
Following proposition \ref{propos:passive_sep_and_poly_factorization}, this implies that 
\begin{equation}\label{eq:two-mode-passive_separable_polynomials}
   p_{C}( \mathbf{U}^T \mathbf{z}) \ = \ p_{1}\left(z_1\right) p_{2}\left(z_2\right),
\end{equation}
where $z_1,z_2\in \complex$. This implies that the algebraic properties of $p_C$ are directly inherited from $p_{1(2)}$. We analyze the surfaces of zeros of this polynomial, \ie the set of solutions for 
\begin{equation}\label{eq:zero-eq-stellar-poly-two-modes}
    p_C(\mathbf z)=0,
\end{equation}
 which in the algebraic geometry jargon define an algebraic variety $V(p_C)$. Based on equation ~\eqref{eq:two-mode-passive_separable_polynomials} we can show the following result 
 \begin{lema}
     \textbf{Zeros of a \textit{passively separable} stellar polynomial} The set of zeros of a stellar polynomial $p_C$, that corresponds to a passively separable state, is given by the union of the zeros of the stellar polynomials $p_1$, and $p_2$, \ie
     \begin{equation*}
         V(p_C(\mathbf z))=V(p_1((\mathbf{U} \mathbf z)_1))\cup V( p_2((\mathbf{U} \mathbf z)_2)).
     \end{equation*}
     \begin{proof}
         From equation ~\eqref{eq:two-mode-passive_separable_polynomials} we can deduce that, in any given specific basis 
         \begin{equation*}
             p_C(\mathbf z)=p_1((O^T \mathbf z)_1) p_2((O^T \mathbf z)_2),
         \end{equation*}
         and thus, $p_C(\mathbf z)=0$ for values of $z$ such that $p_1((O^T \mathbf z)_1)=0$ or $p_2((O^T \mathbf z)_2)=0$. This implies that $V(p_C(\mathbf z))=\{\mathbf z| p_C(\mathbf z)=0\}=\{\mathbf z| p_1((O^T \mathbf z)_1)=0 \text{ or } p_2((O^T \mathbf z)_2)=0 \}=\{\mathbf z| p_1((O^T \mathbf z)_1)=0\} \cup \{\mathbf z|p_2((O^T \mathbf z)_2)=0 \}=V(p_1((O^T \mathbf z)_1)) \cup V(p_2((O^T \mathbf z)_2))$
     \end{proof}
 \end{lema}
 A further relevant stage in this analysis is to realize that the varieties of the local polynomials $p_1$ and $p_2$ are \textit{orthogonal} to each other. To prove this we need the following result.
 \begin{lema}\label{lema:transformation_varieties}
 \textbf{Action of passive transformations on varieties:} Let $V(p(\mathbf z))$ be the variety associated with a polynomial $p(\mathbf z)$. Then the variety associated to the polynomial after a mode basis transformation $V(p(O^T \mathbf z))$ is given by 
 \begin{equation*}
     V(p(O^T \mathbf z))=O V(p(\mathbf z)).
 \end{equation*}
\begin{proof}
    The proof is a natural consequence of $V(p(\mathbf z))=\{\mathbf z| p(\mathbf z)=0\}$ being just a collection of hypersurfaces on phase space, which will transform in the same way as any vector on phase space. $V(p(O^T\mathbf z))=\{\mathbf z| p(O^T\mathbf z)=0\}=\{O \mathbf z| p(\mathbf z)=0\}=O V(p(\mathbf z))$.
\end{proof}
 \end{lema}
 Notice that each variety will be typically composed of several parametrized surfaces; thus, we are being a bit sloppy in notation when we represent the action of $O\in SO(2M)$ on $V$, and what we imply is the action of $O$ on each of the parametrized surfaces that compose $V$. \par
 The next result arrives as an immediate consequence of the proof above. 
 \begin{thm}\label{theorem:orthogonality_passively separable_varieties}
     $V(p_1((O^T \mathbf z)_1))\perp V(p_2((O^T \mathbf z)_2))$ and thus, the set of surfaces that compose $V(p( \mathbf z))$ can be divided into two sets of mutually orthogonal surfaces.
     \begin{proof}
         From lemma \ref{lema:transformation_varieties},  $V(p_i((O^T \mathbf z)_i))=O^TV(p_i(\mathbf z_i))$, where $V(p_1(\mathbf z_1))=\bigcup_{k=1}^{r_1}(\alpha_1^{(k)},z_2)$, and  $V(p_2(\mathbf z_2))=\bigcup_{k=1}^{r_2}(z_1,\alpha_2^{(k)})$. Here $r_i=\text{deg}(\mathcal{P}_i)$, and $\alpha_i^{(k)}$ are complex constants corresponding to each zero of the single variable complex polynomials. As a consequence the hypersurfaces of zeros of $p_1$ are parallel to the complex plane $z_2\in \complex$,  and the ones corresponding to $p_2$ are parallel to the complex plane $z_1\in \complex$, and are thus orthogonal to each other. This orthogonality is preserved under the action of a passive transformation. Just for the sake of illustration, which may be useful for later analysis, we build up a specific example. Consider the transformation that we apply is just a real beam splitter with angle $\theta$, so that 
         \begin{equation*}
             \begin{split}
                 &z_1^{'}=\cos \theta z_1 + \sin \theta z_2 \\
                 &z_2^{'}=\cos \theta z_2 - \sin \theta z_1, 
             \end{split}
         \end{equation*}
        then a hyperplane parallel to the $z_2$ complex planes (a zero of $p_1$, and a hyperplane parallel to $z_1$, will transform into 
        \begin{equation}\label{eq:perpendicular_planes}
            \begin{split}
                & (\alpha_1,z_2) \rightarrow (\kappa z_2^{'}+C_1,z_2^{'}) \\
                & (z_1,\alpha_2) \rightarrow (z_1^{'},-\kappa z_1^{'}+C_2),
            \end{split}
        \end{equation}
         where $\kappa=\tan \theta$, and $C_i=\frac{\alpha_i}{\cos \theta}$. Thus, any pair of vectors $\mathbf v_1=(\kappa t_1,t_1)$, $\mathbf v_2=(t_2,-\kappa t_2)$, with $t_i\in\mathbb R$, along each of these planes are orthogonal to each other, \ie $\mathbf v_2 \cdot \mathbf v_1=0$.
         
     \end{proof}
 \end{thm}
 The final consequence of the previous chain of statements is shown below. 
 \begin{thm}\renewcommand{\thethm}{\ref{theorem:orthogonality_is_equivalent_to_passive_separability}}

    A pure core state $|C\rangle$, of a two-mode bosonic system, is passively separable if and only if the hypersurfaces of zeros of its stellar polynomial $p_C(\mathbf z)$ can be separated into two sets of hyperplanes orthogonal to each other, \ie 
    \begin{equation*}
        V(p_C(\mathbf z))=V_1 \cup V_2, \text{ such that } V_1 \perp V_2.
    \end{equation*}
    \begin{proof}
        The proof of this theorem is a direct application of all the results listed above, in a backward order.
    \end{proof}
 \end{thm}

Consider the core state $|C_\psi\rangle$, whose stellar polynomial is $p_{C_\psi}$. We know that if the state is passively separable, its zero surfaces are all planes, and they are orthogonal to each other (or parallel). We can thus test the separability using the following procedure. The equation of interest is 
\begin{equation}
    p_{|C_{\psi}\>}(z_1,z_2)=0,
\end{equation}
We can substitute one of the parametrizations in ~\eqref{eq:perpendicular_planes} so that we get the condition that there exist $\kappa,C \in \complex$, so that 
\begin{equation}\label{eq:ansatz}
    p_{|C_{\psi}\>}(z_1,\kappa z_1+C)=0 \forall z_1.
\end{equation}
Under this ansatz, the planes written in the form $z_1=-\kappa z_2+C$, appear as $z_2=-1/\kappa z_1+C'$. Here, differently from ~\eqref{eq:perpendicular_planes}, we need to allow for arbitrary local phase rotation and as a consequence $\kappa \in \complex$. The condition for two planes to be orthogonals become $\kappa_1^* \kappa_2=-1$. Notice that this implies that all the coefficients $f_{i}$ of $z_1$ in the new polynomial, 
\begin{equation}
    \begin{split}
    p^{'}(z_1) &=p_{|C_{\psi}\rangle}(z_1, \kappa z_1+C)\\
    &=f_0(C)+f_1(\kappa,C) z_1+...+f_{r}z_1^{r},
    \end{split}
\end{equation}
have to be zero. Each of these coefficients is a polynomial of order $N_{max}=r^{*}$, the stellar rank of the state. This simplifies the task, as all possible solutions $\{\kappa_{l},C_{l}\}$ are already encountered after checking the zeroth and first order coefficients $f_0(C)$ and $f_1(\kappa,C)$. The procedure then works as follows 
\begin{enumerate}
    \item Solve $f_0(C)=0$. This is a polynomial on $C$ of order $r^{*}$, so we should get $r^{*}$ possible values of $C$ (counting multiplicities), 
    \item For each value of $C$ obtained in step $1$, solve $f_1(\kappa,C_{k})=0$. This function is linear in $\kappa$, so it provides a single solution $\kappa_{k}$ for each value of $C_{k}$. It may happen that for some values of $C_{\kappa}$ this equation becomes trivial. In that case we need to keep searching for possible solutions at higher order coefficients. If there are more than two different values of $\kappa$ we can already be sure that the state is not passively separable. 
    \item Finish the computation of pairs $(\kappa_k,C_k)$ for the values of $C_k$ that rendered $f_1$ trivial. If there are only two different values of $\kappa$, check that all pairs $(\kappa_k,C_{k})$ satisfy $f_{j}(\kappa_k,C_k)=0$ for $j\geq 2$.  
\end{enumerate}
Notice that the phase of the slopes $\kappa_{I(J)}$ should not be overlooked as it plays a relevant role. If two slopes have the same absolute value but different phase the state is already non separable under passive operations. \par
There is a subtlety still to be considered here. If we happen to be checking in the basis in which the state is actually separable we know that the solutions are actually $z_{2}=C$ or $z_1=C$. The former case is correctly captured by the ansatz, but results in $\kappa=0$. The latter case, on the other hand, has to be checked by just replacing the ansatz by its inverse: $z_{1}=\kappa^{'} z_{2}+C$ and will result in $\kappa^{'}=0$. \par 
\section{Specific examples of two-mode states}\label{sec:appendix_B}

\subsection{Hong-Ou-Mandel state}
A classical example of a passively separable core state is a Hong-Ou-Mandel like state
\begin{equation}
    |C_{HOM}\>=\frac{|02\>-|20\>}{\sqrt{2}}.
\end{equation}
This state can clearly be brought to a factorized form by using a balanced beam splitter. Moreover, it is a subclass of the NOON states, already considered in the main text. We present it here as a simple illustration of the procedure described in Appendix \ref{sec:Appendix_A}. \par 
The stellar function of the state is given by 
\begin{equation}
    P_{\psi_{HOM}}(z_1,z_2)=\frac{z_{1}^{2}-z_{2}^{2}}{2}.
\end{equation}
Though obtaining the solutions directly is easy, we use this example to illustrate the procedure to test the ansatz \eqref{eq:ansatz}. We assume $z_2=\kappa z_1$. Notice that because the polynomial is homogeneous, \ie it has no independent term, the inclusion of a constant $C$ will provide a trivial solution $C=0$ of multiplicity $2$. In the end then only the contributions at second order count, providing:  
\begin{equation}
    P_{\phi_{HOM}}(z_1,\kappa z_1)=z_{1}^{2} \frac{\kappa^2-1}{2}.
\end{equation}
We can observe that the latter is equal to zero for every value of $z_1$ if and only if 
\begin{equation}
    \kappa=\pm 1.
\end{equation}
This implies that the zero surfaces of the polynomial are given by the planes 
\begin{equation}
    z_{2}=\pm z_{1}.
\end{equation}
We can now check that indeed $\kappa_{+}\kappa_{-}=-1$, which implies that the two planes are orthogonal as expected.  \par
It is important to remark that it is not difficult to define the full family of separable states of stellar rank 2. There are only two kinds of core states of stellar rank 2 that are separable: 
\begin{enumerate}
    \item Hong-Ou-Mandel-like states: $\hat{U}_{O}|11\>$,
    \item One-sided states:$\hat{U}_{O}(a |0\>+b|1\>+c|2\>)\otimes |0\>$,
\end{enumerate}
where $\hat{U}$ is an arbitrary passive unitary.

\subsection{Maximally entangled core states}
 We consider now maximally entangled core states of stellar rank $2N$, given by
 \begin{equation}
    |C_{max-ent}\>=\frac{|00\>+|11\>+...+|NN\>}{\sqrt{N+1}}.
\end{equation}
The stellar function is given by 
\begin{equation}
    P_{max-ent}(z)=\frac{1+z_1 z_2+...+\frac{1}{N}z_{1}^{N}z_{2}^{N}}{\sqrt{N}}.
\end{equation}
If we check the ansatz $z_{2}=\kappa z_1$ we get 
\begin{equation}
    P_{max-ent}(z_1,\kappa z_1)=\frac{1+z_1^2 \kappa+...+\frac{1}{N!}z_{1}^{2N}\kappa^N}{\sqrt{N}}.
\end{equation}
There is no value of $\kappa$ for which this polynomial is equal to zero for every value of $z_1$. This implies that, with certainty, the state is not separable under passive operations. \par
This can be further checked explicitly constructing the solutions. If we do a change of variable $w=z_1 z_2$, this is just a polynomial $P(w)$, whose solutions are given by 
\begin{equation}
    z_{1}z_{2}=c_{i}, \forall i\in \left[1,N\right], \text{ where } c_i\in\complex^2.
\end{equation}
meaning that there are $N$ hyperbolas of zeros (\ie $2N$ disconnected surfaces). 
\subsection{Analysis of two photon subtracted states}
The most general two-photon subtracted state of two modes is given by 
\begin{equation}
    |\psi\>=\hat{A}_2\hat{A}_1\hat{G}|00\>,
\end{equation}
where 
\begin{equation}
    \begin{split}
        & \hat{A}_1=\hat{a}_1\cos(\theta_1)+e^{i\phi_1}\sin{\theta_1}\hat{a}_2\\
        & \hat{A}_2=\hat{a}_1\cos(\theta_2)+e^{i\phi_2}\sin{\theta_2}\hat{a}_2\\,
    \end{split}
\end{equation}
and $\hat{G}$ is a general Gaussian unitary, which can be rewritten, using Bloch-Messiah decomposition, as $\hat{G}=\hat{U}(\hat{S}_1\otimes\hat{S}_2)\hat{V}$, so that its action on $|00\>$ is just given by $\hat{U}(\hat{S}_1\otimes\hat{S}_2)|00\>$. If, additionally, we commute the passive unitary $\hat{U}$ through the superpositions of annihilation operators, we just get different superpositions. This implies that the most general states we need to analyze are given by 
\begin{equation}
    |\psi_{2-subtr}\>=\hat{A}_2\hat{A_1}\hat{S}_1\hat{S}_2|00\>,
\end{equation}
where $\hat{S}_1$ and $\hat{S}_2$ implement real squeezings $r_1$ and $r_2$ respectively. The core state of $|\psi_{2-subtr}\>$ is obtained by commuting the squeezing operations through the chain of annihilation operators. This in turn produces a superposition of annihilation and creation operators. Altogether, we get 
\begin{equation}
    \begin{split}
    & |C_{2-subtr}\>= \\
    &(\cosh(r_1)\sinh(r_1)\cos(\theta_1)\cos(\theta_2)\\
    & +e^{i(\phi_1+\phi_2)}\sin(\theta_1)\cosh(r_2)\sin(\theta_2)\sinh(r_2))|00\>\\
    & +\sinh(r_1)\sinh(r_2)(e^{i\phi_1}\sin(\theta_1)\cos(\theta_2)\\
    & +e^{i\phi_2}\sin(\theta_2)\cos(\theta_1))|11\>\\
    & +\frac{\sinh(r_1)^2\cos(\theta_1)\cos(\theta_2)}{\sqrt{2}}|20\>\\
    & +\frac{\sinh(r_2)^2e^{i(\phi_1+\phi_2)}\sin(\theta_1)\sin(\theta_2)}{\sqrt{2}}|02\>.
    \end{split}
\end{equation}
From this, it is easy to get the corresponding polynomial representation. Nevertheless, it is in general not obvious to extract meaningful conditions other than in a one to one basis, for arbitrary values of the parameters. For that reason we will restrict to the case $\phi_1=\phi_2=0$.
Then we will proceed with checking the ansatz \eqref{eq:ansatz} for different values of $\theta_1,\theta_2$.\par

The polynomial, for $\phi_1=\phi_2=0$, becomes
\begin{equation}
    \begin{split}
    & P_{C_{2-subtr}}(z_1,z_2)= \cosh(r_1)\sinh(r_1)\cos(\theta_1)\cos(\theta_2)\\ &+\sin(\theta_1)\cosh(r_2)\sin(\theta_2)\sinh(r_2)\\
    & +\sinh(r_1)\sinh(r_2)\left(\sin(\theta_1)\cos(\theta_2)+\sin(\theta_2)\cos(\theta_1)\right)z_1 z_2\\
    & +\sinh(r_1)^2\cos(\theta_1)\cos(\theta_2)z_1^2\\
    & +\sinh(r_2)^2\sin(\theta_1)\sin(\theta_2)z_2^2.
    \end{split}
\end{equation}
When we check ansatz \eqref{eq:ansatz}, we can observe that the following configurations of the parameters lead to passively separable states 
\begin{enumerate}
    \item $\theta_1=\theta_2+ n \pi$, for $n\in \mathbb Z$, and for any squeezing values. This implies that the two photons are subtracted from the same mode, as $\hat A_1=\pm \hat A_2$. 
    \item $r_1=r_2$, or $r_1=0$ and $\theta_1=\pm \pi/4 + n\pi$, $\theta_2=\pm \pi/4 +m \pi$, \ie the two photons are subtracted from equally squeezed input states, in balanced superpositions, which can be orthogonal to each other. 
\end{enumerate}

\end{document}